\documentclass{amsart}
\usepackage{amssymb}
\usepackage{url}
\usepackage{bbm}
\usepackage{graphicx}

\usepackage{geometry}

\newcommand{\comment}[1]{}

\theoremstyle{definition}
\newtheorem{theorem}{Theorem}
\newtheorem{lemma}[theorem]{Lemma}
\newtheorem{definition}[theorem]{Definition}

\newtheorem{corollary}[theorem]{Corollary}

\DeclareMathOperator{\lcm}{lcm}
\DeclareMathOperator{\idtransSym}{trans}

\DeclareMathOperator{\nerveSym}{nerve}
\DeclareMathOperator{\hilbSym}{hilb}
\DeclareMathOperator{\hilbNumSym}{H}
\DeclareMathOperator{\facetsSym}{fac}

\newcommand{\hps}{Hilbert-Poincar\'e series}
\def\cocoa{{\hbox{\rm C\kern-.13em o\kern-.07em C\kern-.13em o\kern-.15em A}}}

\newcommand{\N}{\mathbb{N}}

\newcommand{\trans}[1]{\ensuremath{{#1}^T}}
\newcommand{\p}{\ensuremath{^\prime}}
\newcommand{\pp}{\ensuremath{^{\prime\prime}}}
\newcommand{\biimp}{\Leftrightarrow}

\newcommand{\defeq}{\stackrel{\text{\tiny def}}{=}}

\newcommand{\setBuilder}[2]{\left\{{#1}\left|{#2}\right.\right\}}

\newcommand{\idealBuilder}[2]{\left\langle#1\left|#2\right.\right\rangle}

\newcommand{\card}[1]{\left|#1\right|}

\newcommand{\varProd}{\mathbbm{x}}

\newcommand{\proofPart}[1]{{\bf $\boldsymbol{#1}$:}}

\newcommand{\hilb}[1]{\hilbSym\left({#1}\right)}
\newcommand{\hilbNum}[1]{\hilbNumSym\left({#1}\right)}

\newcommand{\ming}[1]{\min\left({#1}\right)}
\newcommand{\ideal}[1]{\left<#1\right>}
\newcommand{\set}[1]{\left\{{#1}\right\}}
\newcommand{\facets}[1]{\facetsSym\left({#1}\right)}
\newcommand{\koz}[2]{\Delta_{#2}^{#1}}
\newcommand{\euler}[1]{\tilde\chi\left({#1}\right)}
\newcommand{\vertices}[1]{V_{#1}}
\newcommand{\NP}{\ensuremath{\text{\tt NP}}}
\newcommand{\NPhard}{\NP\ensuremath{\text{-hard}}}
\newcommand{\NPc}{\NP\ensuremath{\text{-complete}}}
\newcommand{\coNP}{\ensuremath{\text{\tt co-NP}}}
\newcommand{\coNPhard}{\coNP\ensuremath{\text{-hard}}}

\newcommand{\sP}{\ensuremath{\text{\#\tt P}}}
\newcommand{\sPc}{\ensuremath{\sP{}\text{-complete}}}
\newcommand{\sPhard}{\ensuremath{\sP{}\text{-hard}}}
\newcommand{\SAT}{\ensuremath{\text{\tt SAT}}}
\newcommand{\satCount}{{\tt $\#$\SAT{}}}
\newcommand{\countOf}[1]{\#\left(\textrm{#1}\right)}
\newcommand{\comple}[1]{\overline{#1}}
\newcommand{\pows}[1]{\mathcal{P}\left({#1}\right)}
\newcommand{\vsub}{\ominus}

\newcommand{\comtoidSym}{\phi}
\newcommand{\idtocomSym}{\phi^{-1}}
\newcommand{\comtoid}[1]{\comtoidSym\left({#1}\right)}
\newcommand{\idtocom}[1]{\idtocomSym\left({#1}\right)}
\newcommand{\clo}[1]{\ideal{#1}}

\newcommand{\simProd}{\oplus}

\newcommand{\nerve}[1]{\nerveSym\left({#1}\right)}
\newcommand{\EulerChar}{\ensuremath{\text{\tt EulerChar}}}
\newcommand{\IndepSum}{\ensuremath{\text{\tt IndepSum}}}

\newcommand{\Ezero}{\ensuremath{\text{\tt E}_{\boldsymbol 0}}}
\newcommand{\Eeq}{\ensuremath{\text{\tt E}_{\boldsymbol =}}}
\newcommand{\Elt}{\ensuremath{\text{\tt E}_{\boldsymbol <}}}
\newcommand{\Egt}{\ensuremath{\text{\tt E}_{\boldsymbol >}}}
\newcommand{\Egtzero}{\ensuremath{\text{\tt E}_{\boldsymbol >0}}}

\newcommand{\idtomat}[1]{M_{#1}}
\newcommand{\mattoid}[1]{\ideal{#1}}

\newcommand{\idtrans}[1]{\idtransSym\left({#1}\right)}
\newcommand{\pname}[1]{{\tt{#1}}}
\newcommand{\Frobby}{{\tt Frobby}}
\newcommand{\GAP}{{\tt GAP}}
\newcommand{\Sage}{{\tt Sage}}
\newcommand{\Mtwo}{{\tt Macaulay 2}}

\usepackage{threeparttable}

\title{Complexity and Algorithms for Euler Characteristic of
  Simplicial Complexes}

\author{Bjarke Hammersholt Roune\ \ \ \  \and\ \ \  Eduardo S\'aenz-de-Cabez\'on}
\date{\today{}}

\address{\small \rm  Cornell University, \url{http://www.broune.com}}
\email{bhroune@math.cornell.edu}

\address{\small \rm  Universidad de la Rioja, \url{https://belenus.unirioja.es\~esaenz-d}}
\email{eduardo.saenz-de-cabezon@unirioja.es}

\begin{document}
\begin{abstract}
We consider the problem of computing the Euler characteristic of an
abstract simplicial complex given by its vertices and facets. We show
that this problem is {\tt\#P}-complete and present two new practical
algorithms for computing Euler characteristic. The two new algorithms
are derived using combinatorial commutative algebra and we also give a
second description of them that requires no algebra. We present
experiments showing that the two new algorithms can be implemented to
be faster than previous Euler characteristic implementations by a
large margin.
\end{abstract}

\maketitle

\section{Introduction}

The Euler characteristic of a topological space is an invariant used
in a variety of contexts such as category theory, algebraic geometry
and differential geometry. In combinatorics, the Euler characteristic
of a simplicial complex is related to the M\"obius function of a poset
and the inclusion-exclusion principle \cite{S97} and to valuations on
simplicial complexes \cite{KR97} to name but a few connections.

The \emph{reduced Euler characteristic} of an abstract simplicial
complex\footnote{An \emph{abstract simplicial complex} $\Delta$ is a
  family of sets closed under taking subset, so if $\sigma\in\Delta$
  and $\tau\subseteq\sigma$ then $\tau\in\Delta$. This is closely
  related to the notion of a \emph{simplicial complex} which is a set
  of polyhedra with certain properties. All complexes in this paper
  are abstract. See Section \ref{sec:simBackground} for further
  background.} $\Delta$ is
\[
\euler\Delta\defeq-\sum_{\sigma\in\Delta}(-1)^{\card\sigma}=-f_{-1}+f_0-f_1+f_2-f_3+\cdots
\]
where $f_i$ denotes the number of faces (elements) of dimension $i$ in
the complex.\footnote{The usual definition of Euler characteristic is
  $\chi(\Delta)\defeq f_0-f_1+f_2-f_3+\cdots$. The difference is that
  $\chi(\Delta)$ does not count the empty set, while $\euler\Delta$
  does, so $\euler\Delta=\chi(\Delta)-1$. All Euler characteristics in
  this paper are $\euler\Delta$ rather than $\chi(\Delta)$ because
  that simplifies the formulas.} The dimension of a face $\sigma$ is
$\dim(\sigma)\defeq\vert\sigma\vert-1$.

In Section \ref{sec:complexity} we prove that computing the Euler
characteristic of a simplicial complex specified by its vertices and
facets is \sPc{}, which is a formal way of stating that Euler
characteristic is a difficult computational problem. We also show that
the problem of deciding if $\euler\Delta=0$ is not in \NP{} unless
\sP{} is no harder than \NP{}. This answers two open questions posed
by Kaibel and Pfetsch in their survey \cite{simplicialSurvey}.

In Section \ref{sec:algAlg} we introduce two new practical algorithms
for computing Euler characteristic. These two algorithms were
conceived of in terms of combinatorial commutative algebra, and
Section \ref{sec:algAlg} is written solely in terms of
algebra. Section \ref{sec:simAlg} independently describes the same two
algorithms in terms of simplicial complexes and without any reference
to algebra. Section \ref{sec:translation} describes how the algebra
was translated to simplicial complexes and how doing so brought up
interesting mathematics. Finally, Section \ref{sec:bench} presents
experiments that show that the two new algorithms can be implemented
to be faster than previous Euler characteristic implementations by a
large margin.

\section{The Complexity of Euler Characteristic}
\label{sec:complexity}

We describe the complexity class \sP{} and then prove that Euler
characteristic is \sPc{}. This is a precise way of saying that Euler
characteristic is a difficult computational problem. We also consider
the complexity of decision problems associated to Euler
characteristic.  See Section \ref{sec:simBackground} for basic
definitions relating to simplicial complexes.

The complexity of Euler characteristic has been studied before, but
not of a simplicial complex specified by its vertices and facets. It
has been studied for the case of the input being a CW-complex
specified as a circuit \cite{eulerGeom} in the context of real valued
computation and for the input being a sheaf \cite{eulerSheaf} in the
context of algebraic geometry.

\subsection{The complexity class \sP{}}

The complexity class \sP{} is the set of counting problems associated
to decision problems in $\NP$. For example the decision problem ``does
a logical formula have \emph{some} satisfying assignment of
truth-values?''  is in $\NP$, while ``\emph{how many} satisfying
assignments of truth-values does a logical formula have?'' is in
\sP{}. The former is called \SAT{} while the latter is called
\satCount{}. A problem is \sPc{} if it is in \sP{} and any other
problem in \sP{} can be reduced to it in polynomial time.

There is already a list of problems that are known to be \sPc{}, which
is very helpful when proving that a new problem is \sPc{}, as then a
problem in \sP{} is \sPc{} if some other \sPc{} problem reduces to
it. For example it is known that \satCount{} is \sPc{} even when
restricted to formulas with two literals per clause and no negations
\cite{enumerationSP}. A \SAT{} formula is a conjunction of clauses,
where each clause is a disjunction of some number of literals. For
example
\[
(a\lor \lnot b)\land(a\lor c)\land(\lnot b\lor c)
,\]
where $a$, $b$ and $c$ are boolean variables. Here the
satisfying truth assignments $(a,b,c)$ are
\[
\left\{
(0, 0, 1),
(1, 0, 1),
(1, 1, 0),
(1, 1, 1)
\right\}.
\]
The output for \SAT{} with this formula as input is ``yes'' since
there is a satisfying truth assignment. The output for \satCount{} is
``4'', since there are four satisfying truth assignments. The output
for \satCount{} does not include the satisfying truth assignments
themselves, only the number of them.

The program for the rest of this section is to formally define a
problem \EulerChar{} in \sP{} that represents the Euler characteristic
problem, and then to prove that \EulerChar{} is \sPc{}.

\subsection{Euler Characteristic is in \sP{}}
\label{sec:eulerInSP}

The most straigtforward way to define \EulerChar{} would be to have the
input be the facets and vertices of a simplicial complex $\Delta$ and
have the output be simply $\euler\Delta$. It is immediate that this
could never be in \sP{} because $\euler\Delta$ can be negative while
\sP{} is a class of \emph{counting} problems so that their output must
be a natural number.

To arrive at a satisfactory definition of \EulerChar{}, the first step is
to observe that
\begin{equation}
\label{eq:eulerEvenOdd}
\euler\Delta =
\sum_{\sigma\in\Delta}(-1)^{\dim(\sigma)}=
\countOf{odd faces}-\countOf{even faces}.
\end{equation}
Read $\countOf{odd faces}$ as ``the number of odd faces of $\Delta$'',
where a set is odd if it has an odd number of elements which is to say
that its dimension is even. It is not hard to argue that counting the
number of even faces is a problem in \sP{}, and that counting the
number of odd faces is a problem in \sP{} as well. Unfortunately, we
know of no theorem stating that a difference of two functions in \sP{}
is again in \sP{}. So we must find an alternative way to express the
Euler characteristic.

Let $\Delta$ have $n$ vertices. Then
\[
\countOf{even faces}+\countOf{even non-faces}=\countOf{even sets}=2^{n-1},
\]
which together with Equation (\ref{eq:eulerEvenOdd}) implies that
\[
\euler\Delta+2^{n-1} =
\countOf{odd faces}+\countOf{even non-faces}.
\]
Consider the decision problem ``does $\Delta$ have an odd face or an
even non-face?''. This problem is in \NP{} where a certificate of a
``yes''-answer is any concrete odd face or even non-face. Define
\EulerChar{} to be the counting version of this. The input is then the
vertices and facets of $\Delta$ and the output is the number of odd
faces and even non-faces, that is the output is
$\euler\Delta+2^{n-1}$. We can subtract $2^{n-1}$ in polynomial time,
which justifies that \EulerChar{} represents the problem of computing the
Euler characteristic of a simplicial complex.

We conclude that computing Euler characteristic is a problem in \sP{}
when expressed formally in the form of the \EulerChar{} problem.

\subsection{Euler Characteristic is \sP{}-complete}
The main result in this section is Theorem \ref{thm:eulerSPC} which
states that \EulerChar{} is \sPc{}. This is an example of the fact that
even trivial problems can have a counting version that is \sPc{}. To
see that the problem ``does $\Delta$ have an odd face or an even
non-face'' is especially trivial, observe that only $\Delta=\emptyset$
fails to have the even face $\emptyset$.

\begin{theorem}
\label{thm:eulerSPC}
\EulerChar{} is \sPc{}. That is, the problem of computing the Euler
cha\-rac\-te\-ris\-tic of a simplicial complex given by its vertices
and facets is \sPc{}.
\end{theorem}
\begin{proof}
We proved in Section \ref{sec:eulerInSP} that \EulerChar{} is in
\sP{}. We prove the statement of the theorem by showing that the
\sPc{} problem \satCount{} reduces to \EulerChar{}. We introduce an
intermediate problem \IndepSum{} and prove that \satCount{} reduces to
\IndepSum{} and then that \IndepSum{} reduces to \EulerChar{}. Given a \SAT{}
formula $S$, the combination of these two reductions yields a
simplicial complex $\Delta$ such that $\euler\Delta$ is the number of
truth assignments that satisfy $S$.

We need to introduce some terminology. Let the \emph{parity sum} of a
set of sets $S$ be $P(S)\defeq\sum_{s\in S}(-1)^{\card s}$. For
example the parity sum of a simplicial complex $\Delta$ is
$P(\Delta)=-\euler\Delta$. Let $G$ be a simple graph with vertex set
$V$. Then a set of vertices $S\subseteq V$ is \emph{dependent} if it
contains both endpoints of some edge of $G$. Otherwise $S$ is
\emph{independent}.

We can now define the problem \IndepSum{}. The input of \IndepSum{} is the
vertices and edges of a graph $G$, and the output is the parity sum of
the set of independent sets of $G$.

\proofPart{\text{\IndepSum{} reduces to \EulerChar{}}} Let $G$ be a simple
graph with vertex set $V$ and define a simplicial complex $\Delta$
such that the facets of $\Delta$ are the complements of the edges of
$G$. Then a set of vertices is a face of $\Delta$ if and only if the
complement contains an edge, that is if and only if the complement is
a dependent set of $G$.

If $\sigma\subseteq V$ then let $\comple \sigma\defeq V\setminus
\sigma$ be its complement. Let $D$ be the set of dependent sets of $G$
or equivalently $D=\setBuilder{\comple
  \sigma}{\sigma\in\Delta}$. Since $(-1)^{\card{\comple
    \sigma}}=(-1)^n(-1)^{\card \sigma}$ and $-\euler\Delta$ is the
parity sum of the faces of $\Delta$, we get that
\[
P(D)=(-1)^nP(\Delta)=-(-1)^n\euler\Delta.
\]

Let $I$ be the set of independent sets. Every set is either dependent
or independent, so assuming that $V\neq\emptyset$ we get that ($\pows
V$ is the set of all subsets of $V$)
\[
P(I)+P(D)=P(I\cup D)=P(\pows V)=0.
\]
We conclude that $P(I)=-P(D)=(-1)^n\euler\Delta$ so that \IndepSum{}
reduces to \EulerChar{}.

\proofPart{\text{\satCount{} reduces to \IndepSum{}}} Let $S$ be a \SAT{}
formula. We construct a graph $G$ such that the number of truth
assignments that satisfy $S$ equals $(-1)^n$ times the parity sum of
the independent sets of $G$.

\begin{figure}[!b]
\centering
\includegraphics[width=0.9\textwidth]{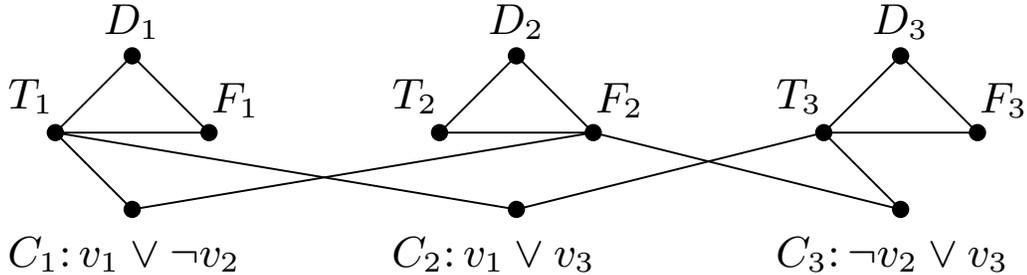}
\caption{Illustration for the proof of Theorem \ref{thm:eulerSPC}.}
\label{fig:satEulerGraph}
\end{figure}

Let $v_1,\ldots,v_n$ be the variables that appear in the formula $S$
and let $c_1,\ldots,c_k$ be the clauses that appear in $S$.  For each
variable $v_i$ we introduce a 3-clique with vertices $T_i$, $F_i$ and
$D_i$. Here $T_i$ represents $v_i$ having the value true and $F_i$
represents false. For each clause $c_j$ we introduce a vertex
$C_j$. If the literal $v_i$ appears in clause $c_j$ with no negation,
then we add an edge between $T_i$ and $C_j$. If the literal $\lnot
v_i$ appears in clause $c_j$ then we add an edge between $F_i$ and
$C_j$. We claim that the number of truth assignments that satisfy $S$
equals $(-1)^n$ times the parity sum of the independent sets of this
graph $G$.

For concreteness, consider the \SAT{} formula
\[
(v_1\lor \lnot v_2)\land(v_1\lor v_3)\land(\lnot v_2\lor v_3).
\]
The graph that we construct based on this formula is shown in Figure
\ref{fig:satEulerGraph}.

Let $A$ be the set of vertices named $D_i$ or $C_j$ and let $B$ be the
set of vertices named $T_i$ or $F_j$. Let $I$ be the set of
independent sets of $G$ and let $I_B$ be the set of independent sets
that are subsets of $B$. Define the function $p\colon I\rightarrow
I_B$ by $p(d)\defeq d\setminus A$.

We are going to prove that $i)$ if $p^{-1}(d)=\set{d}$ then $\card
d=n$ and $ii)$ that the set of such $d$ is in bijection with the set
of truth assignments that satisfy $S$. We are also going to prove
$iii)$ that if $p^{-1}(d)\neq\set d$ then the parity sum
$P(p^{-1}(d))$ is zero. These three statements imply that
\[
\countOf{satisfying truth assignments} =
(-1)^n\sum_{d\in I_B}P(p^{-1}(d)) =
(-1)^nP(I),
\]
where we use that $\set{p^{-1}(d)}_{d\in I_B}$ is a partition of
$I$. It only remains to prove $i)$, $ii)$ and $iii)$.

$i)$ \proofPart{\text{If $p^{-1}(d)=\set d$ then $\card d=n$}}
Suppose that $d\in I_B$ such that $p^{-1}(d)=\set d$. Pick
   some variable $v_i$. Then $d\cup\set{D_i}$ is dependent since
   otherwise it would be an element of $p^{-1}(d)$. As $D_i$ is only
   adjacent to $T_i$ and $F_i$, it must be the case that $d$ contains
   one of $T_i$ and $F_i$. It cannot contain both as there is an edge
   between them. If $d$ contains $T_i$ then we assign the value true
   to $v_i$ and otherwise $d$ contains $F_i$ and we assign the value
   false to $v_i$. In this way $d$ encodes a truth assignment to the
   variables of the formula $S$.

$ii)$ \proofPart{\setBuilder{d\in I_b}{p^{-1}(d)=\set d}\text{ is in
       bijection with the satisfying truth assignments of $S$}} Pick
   some clause $c_j$. Then $d\cup\set{C_j}$ is dependent so $d$ must
   contain some $T_i$ or $F_i$ that is adjacent to $C_j$ and this
   implies that the truth assignment that $d$ represents satisfies the
   clause $c_j$. This establishes a bijection between the set of $d\in
   I_B$ such that $p^{-1}(d)=\set d$ and the set of truth assignments
   that satisfy $S$.

$iii)$ \proofPart{\text{if $p^{-1}(d)\neq\set d$ then
       $P(p^{-1}(d))=0$}} Let $d\in I_b$ such that $p^{-1}(d)\neq\set
   d$. Then we can pick some vertex $a\in A$ such that $d\cup\set a\in
   p^{-1}(d)$. Then $d$ does not contain any vertex that is adjacent
   to $a$, and there are no edges between the elements of $A$, so if
   we let
\[
E\defeq \setBuilder{h\in p^{-1}(d)}{a\notin h}, \quad\quad\quad
F\defeq \setBuilder{h\in p^{-1}(d)}{a\in h}
\]
then $h\mapsto h\cup\set a$ is a bijection from $E$ to $F$ so that
$P(F)=-P(E)$. As $\set{E,F}$ is a partition of $p^{-1}(d)$ we then
get that $P(p^{-1}(d))=P(E)+P(F)=0$.
\end{proof}

The Euler cha\-rac\-te\-ris\-tic is the alternating sum of the entries
of the $f$-vector, so Euler cha\-rac\-te\-ris\-tic reduces to
$f$-vector. So we get the following result of Kaibel and Pfetsch
\cite{simplicialSurvey} as a corollary.

\begin{corollary}
The problem of computing the $f$-vector of a simplicial complex given
by its vertices and facets is \sP{}-hard.
\end{corollary}

\subsection{Decision Problems}

In this section we investigate the complexity of decision problems
associated to Euler characteristic. Kaibel and Pfetsch pose the open
problem of whether deciding $\euler\Delta=0$ lies in $\NP$
\cite{simplicialSurvey}. Theorem \ref{thm:eulerZero} answers this
question in the negative unless \sP{} is no harder than $\NP$. It is a
central conjecture of computational complexity theory that $\sP{}$ is
harder than $\NP$.

\begin{theorem}
\label{thm:eulerZero}
Let $\Ezero$ be the problem of deciding if $\euler\Delta=0$ where
$\Delta$ is a simplicial complex given by its facets and
vertices. Then $\Ezero$ is \coNPhard{}. Also, $\Ezero$ does not lie in
$\NP{}$ unless $\sP{}$ is no harder than $\NP{}$.
\end{theorem}
\begin{proof}
Let $\Elt$ be as in Lemma \ref{lem:eulerDecide} and assume that
$\Ezero$ is in $\NP$. Then $\Elt$ is in $\NP\cap \coNP{}$ by Lemma
\ref{lem:eulerDecide}. This allows us to compute Euler characteristic
in $\NP\cap \coNP{}$ using binary search. Euler characteristic is
\sPc{} by Theorem \ref{thm:eulerSPC}, so then \sP{} is no harder than
$\NP\cap \coNP{}$.

\proofPart{\Ezero{}\text{ is }\coNPhard{}} Let $S$ be a \SAT{}
formula. The proof of Theorem \ref{thm:eulerSPC} constructs a
simplicial complex $\Delta$ such that $\euler\Delta$ is the number of
satisfying truth assignments to $S$. So the \NPc{} problem \SAT{}
reduces to the decision problem $\euler\Delta\neq 0$. So the
complement of \Ezero{} is \NPhard{}, which implies that \Ezero{} is
\coNPhard{}.
\end{proof}

This leaves an open problem of whether $\Ezero$ lies in $\coNP$, since
Theorem \ref{thm:eulerZero} does not rule that out. If $\Ezero$ does
lie in $\coNP$, it would then be proven that $\NP\neq\coNP$ unless
$\sP$ is no harder than $\NP$ since $\Ezero$ would then lie in $\coNP$
and not in $\NP$. It is an open problem whether $\NP\neq\coNP$.

\begin{theorem}
Let $\Egtzero$ be the problem of deciding if $\euler\Delta>0$ where
$\Delta$ is a simplicial complex given by its facets and
vertices. Then $\Egtzero$ is $\sPhard$.
\end{theorem}
\begin{proof}
Let $\Egt$ and $\Elt$ be as in Lemma \ref{lem:eulerDecide}. The
argument used to prove the equivalence of $\Ezero$ and $\Eeq$ in Lemma
\ref{lem:eulerDecide} also works to show that $\Egtzero$ and $\Egt$
are equivalent. Then in particular both $\Egt$ and $\Elt$ reduce to
$\Egtzero$, so Euler characteristic reduces to $\Egtzero$ using binary
search. Euler characteristic is $\sPc$ by Theorem \ref{thm:eulerSPC}
so then $\Egtzero$ is $\sPhard$.
\end{proof}

\begin{lemma}
\label{lem:eulerDecide}
Consider the following decision problems, where $\Delta$ is a
simplicial complex given by its facets and vertices and $k$ is an
integer,
\begin{align*}
\Ezero\colon\quad\euler\Delta=0,\quad\quad\quad\quad&
\Elt\colon\quad\euler\Delta<k,\\
\Eeq\colon\quad\euler\Delta=k,\quad\quad\quad\quad&
\Egt\colon\quad\euler\Delta>k.
\end{align*}
If any one of these problems are in $\NP$ then they are all in
$\NP\cap \coNP{}$. There are polynomial time reductions in both
directions between $\Ezero$ and $\Eeq$ and between $\Elt$ and $\Egt$.
\end{lemma}
\begin{proof}

\proofPart{\Ezero\text{ and }\Eeq\text{ are equivalent}} Assume without
loss of generality that $\Delta\neq\emptyset$. Use Lemma
\ref{lem:consEuler} to construct a polynomial size simplicial complex
$\Gamma$ such that $\euler\Gamma=k-1$ and
$\Delta\cap\Gamma=\set{\emptyset}$. Then $\Psi\defeq\Delta\cup\Gamma$
has $\euler\Psi=\euler\Delta+\euler\Gamma+1=\euler\Delta+k$ so
$\euler\Delta=0$ if and only if $\euler\Psi=k$. This gives a
polynomial time reduction in both directions between $\Ezero$ and $\Eeq$
and also shows that $\Ezero$ is in $\NP$ or $\coNP{}$ if and only if $\Eeq$
is in $\NP$ or $\coNP{}$ respectively.

\proofPart{\Elt\text{ and }\Egt\text{ are equivalent}} Let $\Gamma$ be
a simplicial complex such that $\euler\Gamma=-1$ and such that the set
of vertices of $\Gamma$ is disjoint from the set of vertices of
$\Delta$. Let $\Psi\defeq\Delta\simProd\Gamma$ as in Theorem
\ref{thm:simProd} whereby $\euler\Psi=-\euler\Delta$. Then
$\euler\Delta<k$ if and only if $\euler\Psi>-k$. This gives a
polynomial time reduction in both directions between $\Elt$ and $\Egt$
and also shows that $\Elt$ is in $\NP$ or \coNP{} if and only if
$\Egt$ is in \NP{} or \coNP{} respectively.

\proofPart{\Eeq\text{ in \NP{}}\biimp\Elt\text{ in \NP{}}} 
If $\Eeq$ is in \NP{} then we can certify the exact value of
$\euler\Delta$, which will also serve as a certificate for
$\euler\Delta>k$. If $\Elt$ is in \NP, then so is $\Egt$ in which case
we can certify that $k-1<\euler\Delta<k+1$ which serves as a
certificate of $\euler\Delta=k$.

\proofPart{\text{If one problem is in \NP{}, then they are all in
    \coNP{}}} Assume that one of \Ezero{}, \Eeq{}, \Elt{} and \Egt{}
is in \NP{}. Then we have shown that they are all in \NP{}. So we know
that \Eeq{} is in \NP{}, which allows us to certify the exact value of
$\euler\Delta$. This also serves as a certificate for when
$\euler\Delta<k$ is not true and when $\euler\Delta=k$ is not true, so
$\Elt$ and $\Eeq$ are in \coNP{}. We have already proven that this
implies that $\Egt$ and $\Ezero$ are also in \coNP{}.
\end{proof}

Lemma \ref{lem:consEuler} constructs a simplicial complex with a given
Euler characteristic $k$ such that the bit size of the complex is
bounded by a fixed polynomial in the bit size of the Euler
characteristic which is $\lceil\log k\rceil$. It is necessary to bound
the bit size of the simplicial complex in this way since otherwise the
proof of Lemma \ref{lem:eulerDecide} would not go through.

\begin{lemma}
\label{lem:consEuler}
Let $k$ be an integer. Then there is a simplicial complex $\Delta$
such that $\euler\Delta=k$ and $\Delta$ has no more facets and no more
vertices than $2l^2+3l+7$ where $l=\lceil\log_2(|k|)\rceil$ or $l=0$
if $k=0$.
\end{lemma}
\begin{proof}
The proof is based on inclusion-exclusion along with Theorem
\ref{thm:simProd}.

\proofPart{\text{The case }k=0,1} Let $\Delta_0\defeq\emptyset$ and
$\Delta_1\defeq\ideal{\set{t_1},\set{t_2}}$.

\proofPart{\text{The case }k=2^n} Let $n$ be a positive integer and
define $\Gamma_i\defeq\ideal{\set{x_{ni1}},\set{x_{ni2}},
  \set{x_{ni3}}}$. Let
$\Psi_n\defeq\Gamma_1\simProd\cdots\simProd\Gamma_n$. We have
$\euler{\Gamma_i}=2$ so $\euler{\Psi_n}=2^n$. Also $\Psi_n$ has $3n$
vertices and $3n$ facets.

\proofPart{\text{The case }k>0} Write $k=(b_l\cdots b_0)_2$ in binary
such that $k=\sum_{n=0}^lb_n2^n$, $b_n\in\set{0,1}$ and $b_l\neq
0$. This implies that $l=\lceil\log_2(k)\rceil$.

Let $W$ be a finite set of non-empty simplicial complexes with
disjoint vertex sets. Then the only face in more than one element of
$W$ is $\emptyset$, so $\euler{\cup W}=\sum_{A\in W}\euler A+\card
W-1$. So for $W_n\p\defeq\setBuilder{\Psi_n}{b_n=1}$ we have
$\euler{\cup W_n\p}=k+\card{W_n\p}-1$. Let $p\defeq\card{W\p_n}+1$ and
\[
\Phi\defeq
\ideal{\set{y_1}\,\ldots,\set{y_{p}}}\simProd
\ideal{\set{a,b},\set{a,c},\set{b,c}}.
\]
Then $\euler\Phi=(p-1)*(-1)=-\card{W_n\p}$ by Theorem
\ref{thm:simProd}. Observe that $\Phi$ has no more facets and no more
vertices than $\card{W\p_n}+4\leq l+4$. Let $W_n\defeq
W_n\p\cup\set{\Phi}$ and $\Delta_k\defeq\cup W_n$. Then
\[
\euler{\Delta_k}=
\euler{\cup W_n}=
(k+\euler\Phi)+(\card{W_n\p}+1)-1 = k.
\]
Observe that $\Delta_k$ has no more vertices and no more facets than
\[
(l+4)+\sum_{n=0}^l 3n=l+4+\frac{3}{2}(l(l+1))
.\]

\proofPart{\text{The case }k<0} Let
$\Omega\defeq\ideal{\set{z_1,z_2},\set{z_1,z_3},\set{z_2,z_3}}$ and
observe that $\euler\Omega=-1$. Let
$\Delta_k\defeq\Omega\simProd\Delta_{-k}$ so that
$\euler{\Delta_k}=-\euler{\Delta_{-k}}=k$. Observe that $\Delta_k$ has
no more vertices and no more facets than
\[
l+4+\frac{3}{2}(l(l+1))+3=\frac{3}{2}l^2+\frac{5}{2}l+7\leq 2l^2+3l+7.
\qedhere\]
\end{proof}

\section{Algebraic Algorithms for Euler Characterisic}
\label{sec:algAlg}

In this section we describe two new algorithms for computing Euler
characteristic of a simplicial complex using algebraic techniques. In
Section \ref{sec:simAlg} we present these same two algorithms in the
language of simplicial complexes. This section is independent from
Section \ref{sec:simAlg} and it only uses algebra.

More precisely the two algorithms we present in this section compute
the coefficient of $\varProd\defeq x_1\cdots x_n$ in the multigraded
\hps{} numerator $\hilbNum I$ of a square free monomial ideal $I$. In
Section \ref{sec:translation} we show that this is equivalent to
computing the Euler characteristic of a simplicial complex. For that
reason we define
\[
\euler I \defeq \textrm{coefficient of $\varProd$ in $\hilbNum I$}.
\]
The summary of what Section \ref{sec:translation} shows in detail is
that given a simplicial complex $\Delta$ we can define a monomial
ideal $I$ such that $\euler I=\euler\Delta$. Computing $I$ from
$\Delta$ takes little time. See Section \ref{sec:translation} for
details on the relationship between the algebraic algorithms in this
section and the simplicial algorithms in Section \ref{sec:simAlg}.

\subsection{Background and Notation}

We work in a polynomial ring $\kappa[x_1,\ldots,x_n]$ over a field
$\kappa$ and with variables $x_1,\ldots,x_n$. A \emph{monomial ideal}
is a polynomial ideal generated by monomials. Let $I$ be a monomial
ideal. Each monomial ideal has a unique minimal set of monomial
generators $\ming I$. The \emph{exponent vector} of a monomial $m$ is
a vector $v$ such that $m=\prod_{i=1}^n x_i^{v_i}$. A monomial has
full support if it is divisible by $\varProd$. A monomial ideal has
full support if $\lcm(\ming I)$ has full support. The colon of two
monomials $a$ and $b$ is $a:b\defeq \frac{\lcm(a,b)}{b}$. The colon of
a monomial ideal $I$ by a monomial $a$ is
\[
I:a=\setBuilder{m}{ma\in I}=\idealBuilder{m:a}{m\in I\text{ and $m$ is a
    monomial}}.
\]

The \emph{multigraded \hps{}} $\hilb I$ is the possibly infinite sum
of standard monomials of $I$, that is $ \hilb I\defeq\sum_{m\notin I}m
$ where the sum is taken over monic monomials $m$. The multigraded
\hps{} can be written as a rational function
\[
\hilb I=\frac{\hilbNum I}{(1-x_1)\cdots(1-x_n)}
\]
where $\hilbNum I$ is a polynomial called the \emph{multigraded \hps{}
  numerator}.

\subsection{Divide...}

Both algorithms we present are divide-and-conquer algorithms. They
take a monomial ideal $I$ and split it into two simpler monomial
ideals $J$ and $K$ such that $\euler I=\euler J+\euler K$. This
process proceeds recursively until all the remaining ideals are simple
enough that they can be processed directly.

Let $p$ be a square free monomial and let $I$ be a square free
monomial ideal. The divide steps for the two algorithms are derived
from the equation
\begin{equation}
\label{eqn:hps}
\hilb I = \left(\hilb{I:p}\right)p + \hilb{I + \ideal p}.
\end{equation}
By giving these three \hps{} the same denominator
$(1-x_1)\cdots(1-x_n)$, we get a similar equation for the \hps{}
numerators
\[
\hilbNum I = \left(\hilbNum{I:p}\right)p + \hilbNum{I + \ideal p}.
\]
By considering the coefficient of $\varProd$ on both sides, we then
get that
\[
\euler I = \euler{(I:p)p} + \euler{I + \ideal p}.
\]
It would simplify this expression if we could write $\euler{I:p}$
instead of $\euler{(I:p)p}$. This does not work directly since in
general $\euler{I:p}$ will be zero since $I:p$ will not have full
support. We are working within a polynomial ring
$R\defeq\kappa[x_1,\ldots,x_n]$. We will argue that if we change the
ring that $I:p$ is embedded in, then it becomes true that
$\euler{I:p}=\euler{(I:p)p}$.

Since the variables that divide $p$ do not appear in $\ming{I:p}$ we
can embed $I\p\defeq I:p$ into a ring $R\p\defeq\kappa[P]$ where
$P\defeq\setBuilder{x_i}{x_i\ \text{does not divide}\ p}$. Let
$\varProd\p\defeq\Pi P=\frac{\varProd}{p}$ be the product of the
variables in $R\p$. Then $\euler{I\p}=\euler{(I:p)p}$ so we consider
$I:p$ to be embedded in $R\p$, which gives us the final equation
behind splitting
\begin{equation}
\label{eqn:alg}
\euler I = \euler{I:p} + \euler{I + \ideal p}.
\end{equation}
There are two different algorithms for \hps{} that are based on
Equation \eqref{eqn:hps}. We present two analogous algorithms for
Euler characteristic that are based on Equation \eqref{eqn:alg}.

The \hps{} algorithm due to Bigatti, Conti, Robbiano and Traverso
\cite{bigattiEtAlHilbSerlg, bigattiHilbSerComp} uses Equation
\eqref{eqn:hps} directly as we have written it. It is a
divide-and-conquer algorithm that splits a monomial ideal $I$ into the
two simpler monomial ideals $I:p$ and $I+\ideal p$. We call this
algorithm the \emph{BCRT algorithm for \hps{}}. In analogy with that
algorithm, we propose a BCRT algorithm for Euler characteristic that
uses Equation \eqref{eqn:alg} directly as written -- it splits $I$
into the two simpler ideals $I:p$ and $I+\ideal p$. We call it the
\emph{algebraic BCRT algorithm for Euler characteristic}.

We call the $p$ in Equation \eqref{eqn:alg} the \emph{pivot}. Section
\ref{sec:algPivots} explores strategies for selecting pivots.

There is also a \hps{} algorithm due to Dave Bayer and Michael
Stillman \cite{hseries} that we will call the \emph{DBMS algorithm for
  \hps{}}. It is based on writing Equation \eqref{eqn:hps} as
\[
\hilb{I + \ideal p} = \hilb I - \hilb{I:p}.
\]
Given an ideal $J$, the idea is to choose $p\in\ming J$ and let
$I\defeq\ideal{\ming J\setminus\set p}$ such that
\[
\hilb J=\hilb{I + \ideal p} = \hilb I - \hilb{I:p}.
\]
In this way $J$ splits into the two simpler ideals $I$ and $I:p$. In
the same way, we can rewrite Equation \eqref{eqn:alg} as
\begin{equation}
\label{eqn:removeGen}
\euler J=\euler{I + \ideal p} = \euler I - \euler{I:p}.
\end{equation}
We propose a DBMS algorithm for Euler characteristic that uses this
equation to split $J$ into $I$ and $I:p$. We call it the
\emph{algebraic DBMS algorithm for Euler characteristic}. Note that
the pivots in the DBMS algorithm are minimal generators of the ideal,
which would not make sense for the BCRT algorithm.

\subsection{... and Conquer}
A square free monomial ideal $I$ is a base case for both algorithms
when Theorem \ref{thm:basecase} applies. Note that the improvements in
Section \ref{sec:algImp} enables further base cases.
\begin{theorem}
\label{thm:basecase}
Let $I$ be a square free monomial ideal. Then
\begin{enumerate}
\item if $I$ does not have full support then $\euler I=0$,
\label{thm:basecase:support}
\item if $I$ has full support and the minimal generators $\ming
  I$ of $I$ are pairwise prime monomials then $\euler
  I=(-1)^{\card{\ming I}}$,
\label{thm:basecase:prime}
\item if $I$ has full support and $\card{\ming I}=2$ then $\euler
I=1$.
\label{thm:basecase:two}
\end{enumerate}
\end{theorem}
\begin{proof}
\proofPart{\eqref{thm:basecase:support}}
All the monomials with non-zero coefficient in $\hilbNum I$ can be
written as $\lcm(M)$ for some $M\subseteq\ming I$. If $I$ does not have
full support then neither does any monomial of the form $\lcm(M)$. Then
$\varProd$ must have a zero coefficient since it has full support.

\proofPart{\eqref{thm:basecase:prime}}
If the elements of $\ming I$ are relatively prime then $\hilbNum
I=\prod_{m\in\ming I}(1-m)$. As the ideal has full support we then get
that $\varProd=\Pi_{m\in\ming I}m$ so the coefficient of $\varProd$ is
$(-1)^{\card{\ming I}}$.

\proofPart{\eqref{thm:basecase:two}} If $\set{a,b}\defeq\ming I$ and
$g\defeq\gcd(a,b)$ then $I=g\ideal{c,d}$ where $c\defeq a:g$ and
$d\defeq b:g$. As $c$ and $d$ are relatively prime by construction, we
get that
\[
\hilbNum I=g\hilbNum{\ideal{c,d}}=g(1-c)(1-d)=cdg-cg-dg+g.
\] So the coefficient of $\varProd=cdg$ in
$\hilbNum I$ is 1.
\end{proof}

These base cases improve on the ones for the \hps{} algorithms in that
they apply more often and can be processed more quickly. For example
the \hps{} algorithms as well have a base case when the elements of
$\ming I$ are relatively prime, but it takes exponential time to
process that base case since $\prod_{m\in\ming I}(1-m)$ can have
$2^{\card{\ming I}}$ terms. Here all that is required is to determine
if $\card{\ming I}$ is even or odd. The base case when $I$ does not
have full support does not exist for the \hps{} algorithms.

\subsection{Termination and Complexity}

It is clear that the DBMS algorithm terminates since $\card{\ming I}$
decreases strictly at each step. For the BCRT algorithm, termination
requires that we choose the pivots $p$ such that $1\neq p\notin I$
since otherwise we get an infinite number of steps from $I=I:p$ or
$I=I+\ideal p$.

If we cannot choose a $p$ such that $1\neq p\notin I$ then $I=\ideal
1$ which is a base case.  If $1\neq p\notin I$ and $I$ has full
support, then $I\subsetneq I+\ideal p$ and $I\subsetneq I:p$. So if
the BCRT algorithm does not terminate then there would be an infinite
sequence of strictly increasing ideals in contradiction to the fact
that the ambient polynomial ring is Noetherian. So both algorithms
terminate.

We have seen that the DBMS algorithm gets rid of at least one minimal
generator at each step. The BCRT algorithm gets rid of at least one
variable at each step if the pivot is chosen to be a single variable
$x_i$. It is immediate that $x_i$ is not a variable of $I:x_i$. To see
that $x_i$ can also be removed from $I+\ideal{x_i}$, observe that the
only minimal generator that is divisible by $x_i$ is $x_i$ itself, so
$x_i$ is $(I+\ideal{x_i})$-independent from the other variables and so
can be removed using the independent variables technique from Section
\ref{sec:algImp}.

Since the Euler characteristic problem is \sPc{} we expect both
algorithms to run in at least single exponential time. Using the
transpose technique from Section \ref{sec:algImp}, we can interchange
the number of variables $n$ with the number of minimal generators
$\card{\ming I}$. So if $l\defeq\min(n,\card{\ming I})$ then both the
BCRT and DBMS algorithms have $O(q2^l)$ asymptotic worst case time
complexity where $q$ is a polynomial. So both algorithms run in single
exponential time. We expect that a more careful analysis could reduce
the base of the exponential.

\subsection{Pivot Selection}
\label{sec:algPivots}

We have proven that the BCRT and DBMS algorithms terminate in a finite
amount of time. To be useful in practice, the amount of time until
termination should be small rather than just finite. The strategy used
for selecting pivots when splitting an ideal has a significant impact
on performance. We describe several different pivot selection
strategies here and compare them empirically in Section
\ref{sec:bench}.

A \emph{popular variable} is a variable that divides a maximum number
of minimal ge\-ne\-ra\-tors of the ideal. In other words, a popular
variable $x_i$ maximizes $\card{\ming I\cap\ideal{x_i}}$. A \emph{rare
  variable} is a variable $x_i$ that minimizes $\card{\ming
  I\cap\ideal{x_i}}$ with the constraint that $x_i\notin\ming I$.

If there are several candidate pivots that fit a given pivot selection
strategy, then the pivot used is chosen in an arbitrary deterministic
way among the tied candidates.

\subsubsection*{BCRT Pivot Selection Strategies}
Recall that BCRT pivots $p$ satisfy $1\neq p\notin I$.

\begin{description}
\item[popvar] The pivot is a popular variable.
\item[rarevar] The pivot is a rare variable.
\item[random] The pivot is a random variable $e$ such that $e\notin\ming I$.
\item[popgcd] Let $x_i$ be a popular variable. The pivot is the gcd of
  three minimal generators chosen uniformly at random among those
  minimal generators that $x_i$ divides.
\end{description}

The strategies \pname{popvar} and \pname{popgcd} have been found to
work well for the BCRT algorithm for \hps{}, so we also try them
here. \pname{rarevar} and \pname{random} have been included to have
something to compare to.

\subsubsection*{DBMS Pivot Selection Strategies}

Recall that DBMS pivots are elements of $\ming I$.

\begin{description}
\item[rarevar] The pivot is a minimal generator divisible by a rare variable.
\item[popvar] The pivot is a minimal generator divisible by a popular variable.
\item[maxsupp] The pivot is a minimal generator with maximum support.
\item[minsupp] The pivot is a minimal generator with minimum support.
\item[random] The pivot is a minimal generator chosen uniformly at random.
\item[rarest] The pivot is a generator that is divisible by a
  maximum number of rare variables. Break ties by picking the
  generator that is divisible by the maximum number of
  second-most-rare variables and so on.
\item[raremax] The pivot is chosen according to \pname{rarevar} where
  ties are broken according to \pname{maxsupp}.
\end{description}

\subsection{Improvements}
\label{sec:algImp}

We present several improvements to the DBMS and BCRT algorithms.

\subsubsection*{Independent Variables}

We say that two subsets $A,B\subseteq\set{x_1,\ldots,x_n}$ are
\emph{$I$-independent} if $\ming I$ is the disjoint union of
$\ming{I\cap\kappa[A]}$ and $\ming{I\cap\kappa[B]}$. This is another
way of saying that the minimal generators of $I$ can be partitioned
into two subsets such that every generator in one set is relatively
prime to every generator from the other set. If $A$ and $B$ are
$I$-independent then
\[
\hilb I=\hilb{I\cap\kappa[A]}\cdot\hilb{I\cap\kappa[B]}.
\]
This is a standard technique for computing \hps{}
\cite{hseries,bigattiEtAlHilbSerlg,rouneSliceIrrDecom}, and it applies
to computing Euler characteristic as well since
\[
\euler I=\euler{I\cap\kappa[A]}\cdot\euler{I\cap\kappa[B]}.
\]
The existence of an $I$-independent pair $(A,B)$ can be determined in
nearly linear time \cite{rouneSliceIrrDecom}, and such a pair can cut
down on computation time dramatically. However, independence rarely
occurs for large random ideals and detecting it does take some time,
so this technique is not worth it unless there is some reason to
suspect that it will apply to a given ideal.

\subsubsection*{Eliminate Unique Variables}

Suppose $x_i$ divides only one minimal generator $p\in\ming J$. Use
Equation \eqref{eqn:removeGen} with $p$ as the pivot to get that
\[
\euler J=\euler{I + \ideal p} = \euler I - \euler{I:p} = -\euler{I:p},
\]
since $\euler I=0$ as no generator of $I$ is divisible by $x_i$ so $I$
does not have full support.

\subsubsection*{Transpose Ideals}

Let $M$ be a matrix whose entries are 0 or 1. Each row of $M$ is then
a 0-1 vector that we can interpret as the exponent vector of a square
free monomial. Define $\mattoid M$ to be the monomial ideal generated
by the monomials whose exponent vectors are the rows of $M$. On the
other hand, given a square free monomial ideal $I$, we can take the
exponent vectors of the minimal generators of $I$ and put them in a
matrix. Define $\idtomat I$ to be that matrix. We label the rows of
$\idtomat I$ with the elements of $\ming I$ and we label the columns
of $\idtomat I$ with the variables in the ambient polynomial ring of
$I$. Let the \emph{transpose ideal} $\trans I$ of $I$ be the ideal
generated by the transpose of the matrix of $I$, so $\idtrans
I\defeq\ideal{\trans{\idtomat I}}$.

Theorem \ref{thm:eulerTransId} states that $\euler I=\euler{\idtrans
  I}$, so we can transpose the ideal without changing the Euler
characteristic. The BCRT algorithm is more sensitive to the number of
variables than it is to the number of generators, so it can be
beneficial to transpose the ideal if it has fewer generators than
variables. The DBMS algorithm is opposite of this in that it is more
sensitive to the number of generators than the number of variables.

Another situation where transposing can be beneficial is in the case
where a column of $\idtomat I$ dominates another column. If we take the
transpose, those two columns will become generators and the dominating
generator will not be minimal. When we then take the transpose again,
we will have fewer variables than we started with. This process can
repeat itself several times as the removal of dominating columns from
the matrix can cause rows to be dominating, and removing those
dominating rows can then cause yet more columns to become dominating.

\begin{theorem}
\label{thm:eulerTransId}
If $I$ is a square free monomial ideal then $\euler I=\euler{\idtrans
  I}$.
\end{theorem}
\begin{proof}
The proof is by induction on the number of variables $n$. Let $I$ be a
square free monomial ideal. Choose a variable $x_i$ and let
\[
J\defeq
\ideal{\ming I\setminus\ideal{x_i}} =
\idealBuilder{m\in\ming I}{x_i\text{ does not divide }m}.
\]
The plan of the proof is to show that
\begin{equation}
\label{eqn:trans1}
\euler I = \euler{I:x_i} - \euler{J:x_i}
\end{equation}
and that
\begin{equation}
\label{eqn:trans2}
\euler{\idtrans I} = \euler{\idtrans{I:x_i}} - \euler{\idtrans{J:x_i}}.
\end{equation}
Recall that we embed $I:x_i$ and $J:x_i$ in a subring that does not
have the variable $x_i$, so the result follows from these two
equations by applying the induction assumption to $I:x_i$ and $J:x_i$.
It only remains to prove Equation \eqref{eqn:trans1} and Equation
\eqref{eqn:trans2}.

\proofPart{\text{Equation }\eqref{eqn:trans1}}
Equation \ref{eqn:alg} with $x_i$ as the pivot implies that
\[ 
\euler I = \euler{I:x_i} + \euler{I+\ideal{x_i}}.
\]
Now $J$ does not have full support and $I+\ideal{x_i}=J+\ideal{x_i}$,
so we get by Equation \eqref{eqn:removeGen} that
\[
\euler{I+\ideal{x_i}} =
\euler{J+\ideal{x_i}} = \euler J - \euler{J:x_i} = -\euler{J:x_i}.
\]

\proofPart{\text{Equation }\eqref{eqn:trans2}} This part of the proof
is easier to follow by the reader drawing pictures of the matrices
involved. Let $v$ be column $x_i$ of $\idtomat I$. Then the entries of
$v$ are indexed by $\ming I$ and if $m\in\ming I$ then $v_m=1$ if and
only if $x_i|m$. Let $A$ be the result of removing column $x_i$ from
$\idtomat I$. Then $\ideal{\trans{\idtomat I}} = \ideal{\trans A} +
\ideal{x^v}$ so Equation \eqref{eqn:removeGen} implies that
\[
\euler{\idtrans I} =
\euler{\ideal{\trans{\idtomat I}}} =
\euler{\ideal{\trans A} + \ideal{x^v}} =
\euler{\ideal{\trans A}} - \euler{\ideal{\trans A}:x^v}.
\]
The colon $I:x_i$ corresponds to removing column $x_i$ of $\idtomat I$
so $\ideal A=I:x_i$. The colon can also reduce the number of minimal
generators, so $\idtomat{I:x_i}$ can have fewer rows than $A$
does. However, those extra rows are exponent vectors of non-minimal
generators so they do not impact $\ideal A$. Then Lemma
\ref{lem:domColEuler} implies that
\[
\euler{\idtrans{I:x_i}} =
\euler{\idtrans{\mattoid A}} =
\euler{\ideal{\trans A}}.
\]

It now suffices to prove that $\idtrans{J:x_i}=\ideal{\trans
  A}:x^v$. Let $B$ be the result of removing column $x_i$ of $\idtomat
I$ and also those rows $m\in\ming I$ such that $v_m=1$. We see that
$\idtomat J$ is $\idtomat I$ with the same rows removed as for
$B$. Doing a colon by $x_i$ removes column $x_i$ so $J:x_i=\ideal B$.

It remains to prove that $\ideal{\trans B}=\ideal{\trans A}:x^v$. Observe that
$\ideal{\trans A}:x^v$ removes those columns $m\in\ming I$ of $\trans
A$ where $v_m=1$ which are the same columns that are removed from
$\trans B$. Both $\trans A$ and $\trans B$ have had row $x_i$ removed
so $\ideal{\trans A}:x^v=\ideal{\trans B}$. We have proven that
\[
\idtrans{J:x_i}=\idtrans{\ideal B}=\ideal{\trans B}=\ideal{\trans A}:x^v
\]
where $\idtrans{\ideal B}=\ideal{\trans B}$ depends on the observation
that no row of $B$ dominates any other.
\end{proof}

\begin{lemma}
\label{lem:domColEuler}
If $\ideal A=\ideal B$ then $\euler{\ideal{\trans
    A}}=\euler{\ideal{\trans B}}$ for matrices $A$ and $B$.
\end{lemma}
\begin{proof}
Assume without loss of generality that no row of $B$ dominates any
other. Then $A$ has all the rows that $B$ does and also some
additional non-minimal rows. Assume by induction that there is only
one additional row $r$. Let $d$ be some other row of $A$ that is
dominated by $r$. Then $d$ and $r$ contribute variables $v_d$ and
$v_r$ to the ambient ring of $\ideal{\trans A}$. We get by Equation
\ref{eqn:alg} that
\[
\euler{\ideal{\trans A}}=
\euler{\ideal{\trans A}:v_r}+
\euler{\ideal{\trans A}+\ideal{v_r}}.
\]
All generators of $\ideal{\trans A}$ that are divisible by $v_d$ are
also divisible by $v_r$, so $\ideal{\trans A}+\ideal{v_r}$ does not
have full support at $v_r$ so $\euler{\ideal{\trans
    A}+\ideal{v_r}}=0$. Therefore we have that
\[
\euler{\ideal{\trans A}}=\euler{\ideal{\trans A}:v_r}=\euler{\ideal{\trans B}},
\]
where we are using that the colon $\ideal{\trans A}:v_r$ corresponds
to removing row $r$ from $A$.
\end{proof}

\subsubsection*{Base Case for $\card{\ming I}=3$}

Assume that all unique variables have been eliminated, that $I$ has
full support and that $\card{\ming I}=3$. Then every variable $x_i$
divides 2 or 3 elements of $\ming I$. We can ignore the variables that
divide all 3 minimal generators as they make no difference to the
Euler characteristic. For every minimal generator there must be at
least one variable that does not divide it. So after removing repeated
variables by taking the transpose twice, we see that $I$ must have the
same Euler characteristic as $\ideal{xy,xz,yz}$. So $\euler I=2$.

\subsubsection*{Partial Base Case for $\card{\ming I}=4$}

Suppose that $\card{\ming I}=4$, that every variable divides exactly
two elements of $\ming I$ and that the number of variables is 4. Then
$\euler I=-1$. There should be more rules like this, though
identifying them by hand is laborious and error prone.

\subsubsection*{Make a Table}
It would be beneficial for each small $k$ to perform a computer search
to make a table of all ideals $I$ with $\card{\ming I}=k$ up to
reordering of the variables and the various techniques for simplifying
an ideal that we have presented. Then the Euler characteristic of
ideals with few generators could be computed through a table look-up.

\subsubsection*{Data Structures}
The exponents of $I$ are all 0 or 1, so we can pack 32 or 64 exponents
into a single 32 or 64 bit machine word, and in that representation
many operations become much faster. We used this standard technique in
our implementation and while the general concept is simple we warn
that the implementation details are tricky.

For sparse or complement-of-sparse exponent vectors, it might pay off
to only record the zero entries and one entries respectively, though
this is not something that we have pursued.

\section{Simplicial Algorithms for Euler Characteristic}
\label{sec:simAlg}

In this section we present two algorithms for Euler characteristic
that work directly with simplicial complexes. These two algorithms are
equivalent to the monomial ideal based algorithms from Section
\ref{sec:algAlg}. We introduce the simplicial versions from the ground
up, so this section is independent from Section \ref{sec:algAlg} and
does not use any algebra. See Section \ref{sec:translation} for more
on the connection between the algebraic and simplicial versions of the
two algorithms.

\subsection{Background and Notation}
\label{sec:simBackground}

In the introduction we wrote that a simplicial complex is a finite set
of sets that is closed under subset. We are going to set up an
algebra-simplicial dictionary, and for that to work we associate a set
of vertices to a simplicial complex. This way a simplicial complex can
have vertices that are not an element of any of its faces. Section
\ref{sec:inverseComtoid} shows why this is necessary. Definition
\ref{def:simWithVertex} adds this vertex set structure to simplicial
complexes.

Let $\pows V$ be the set of all subsets of $V$ for any set $V$.

\begin{definition}
\label{def:simWithVertex}
Given a finite set $V$, a simplicial complex $\Delta$ is a subset of
$\pows V$ that is closed under taking subsets. In other words,
$\Delta$ is a set of subsets of $V$ such that if $\sigma\in\Delta$ and
$\tau\subseteq\sigma$ then $\tau\in\Delta$. The elements of $\Delta$
are called \emph{faces} and the elements of $\vertices\Delta\defeq V$
are called \emph{vertices} of $\Delta$.
\end{definition}

The \emph{facets} of $\Delta$ are the maximal faces of $\Delta$ with
respect to inclusion. The set of facets is denoted by
$\facets\Delta$. The \emph{closure} $\clo D$ of a set of sets
$D\subseteq\pows V$ is the smallest simplicial complex that contains
$F$, namely $\clo F\defeq\cup_{d\in D}\pows d$. A simplicial complex
is uniquely given by its facets since
$\Delta=\clo{\facets\Delta}$. The \emph{complement} $\comple\sigma$ of
a set $\sigma\subseteq\vertices\Delta$ is
$\comple\sigma\defeq\vertices\Delta\setminus\sigma$. Two sets
$\sigma,\tau\subseteq\vertices\Delta$ are \emph{co-disjoint} if their
complements are disjoint, or equivalently if
$\sigma\cup\tau=\vertices\Delta$. If $\tau\subseteq\vertices\Delta$
then
\[
\Delta\vsub\tau\defeq
\setBuilder{\sigma\in\Delta}{\sigma\cap\tau=\emptyset} =
\setBuilder{\sigma\setminus\tau}{\sigma\in\Delta}.
\]
The vertex set of $\Delta\vsub\tau$ is
$\vertices{\Delta\vsub\tau}\defeq\vertices\Delta\setminus\tau$.

We will use that Euler characteristic respects inclusion-exclusion in
the sense that
\[
\euler{\Delta\cup\Delta\p}=\euler\Delta+\euler{\Delta\p}-\euler{\Delta\cap\Delta\p}.
\]

\subsection{Divide...}

Both algorithms we present are divide-and-conquer algorithms. They
take a simplicial complex $\Delta$ and split it into two simpler
complexes $\Delta\p$ and $\Delta\pp$ such that
$\euler\Delta=\euler{\Delta\p}+\euler{\Delta\pp}$. This process
proceeds recursively until all the remaining complexes are simple
enough that they can be processed directly. Splitting (the divide
step) is based on Theorem \ref{thm:com}.

\begin{theorem}
\label{thm:com}
If $\sigma$ is a non-empty set of vertices then
\[
\euler\Delta =
 \euler{\Delta\vsub\comple\sigma} +
 \euler{\Delta\cup\pows\sigma}.
\]
\end{theorem}
\begin{proof}
First observe that
\[
\Delta\cap\pows\sigma =
 \setBuilder{\tau\in\Delta}{\tau\subseteq\sigma} =
 \setBuilder{\tau\in\Delta}{\tau\cap\comple\sigma=\emptyset} =
 \Delta\vsub\comple\sigma.
\]
Euler characteristic respects inclusion-exclusion so then
\begin{align*}
\euler{\Delta\cup\pows\sigma} &=
 \euler\Delta +
 \euler{\pows\sigma} -
 \euler{\Delta\cap\pows\sigma} \\&=
 \euler{\Delta} - \euler{\Delta\vsub\comple\sigma},
\end{align*}
where $\euler{\pows\sigma}=0$ as $\sigma$ is non-empty.
\end{proof}

In analogy with the algebraic BCRT algorithm from Section
\ref{sec:algAlg}, we can use the equation in Theorem \ref{thm:com} as
written to split a simplicial complex $\Delta$ into the simpler
complexes $\Delta\vsub\comple\sigma$ and $\Delta\cup\pows\sigma$. We
will refer to this algorithm as the \emph{simplicial BCRT
  algorithm}. It is not immediately clear that
$\Delta\vsub\comple\sigma$ and $\Delta\cup\pows\sigma$ are simpler
than $\Delta$ is. For now, consider that $\Delta\vsub\comple\sigma$
has fewer vertices if $\sigma\subsetneq\vertices\Delta$ and that
$\Delta\cup\pows\sigma$ has fewer non-faces if $\sigma\notin\Delta$.

We call the $\sigma$ in Theorem \ref{thm:com} the
\emph{pivot}. Section \ref{sec:simPivots} explores strategies for
selecting pivots.

We can also write the equation in Theorem \ref{thm:com} as
\[
\euler{\Delta\cup\pows\sigma} =
\euler\Delta - \euler{\Delta\vsub\comple\sigma}.
\]
Let $D$ be a simplicial complex and let $\sigma\in\facets D$. If
$\Delta\defeq\clo{\facets D\setminus\set\sigma}$ then
\begin{equation}
\label{eqn:simSplitFacet}
\euler D =
\euler{\Delta\cup\pows\sigma} =
\euler\Delta - \euler{\Delta\vsub\comple\sigma}.
\end{equation}
In analogy with the algebraic DBMS algorithm from Section
\ref{sec:algAlg}, we can use this equation to split a simplicial
complex $D$ into the simpler complexes $\Delta$ and
$\Delta\vsub\comple\sigma$. We will refer to this algorithm as the
\emph{simplicial DBMS algorithm}. Note that the pivots in the
simplicial DBMS algorithm are facets of the complex, which would not
make sense for the simplicial BCRT algorithm.

\subsection{... and Conquer}
A simplicial complex $\Delta$ is a base case for both algorithms when
Theorem \ref{thm:simBasecase} applies. The improvements in Section
\ref{sec:simImp} enable further base cases.
\begin{theorem}
\label{thm:simBasecase}
Let $\Delta\neq\emptyset$ be a simplicial complex. Then
\begin{enumerate}
\item if $\Delta$ is a cone then $\euler\Delta=0$,
\label{thm:simBasecase:support}
\item if $\Delta$ is not a cone and the facets of $\Delta$ are
  pairwise co-disjoint then $\euler\Delta=(-1)^{\card{\facets I}}$,
\label{thm:simBasecase:prime}
\item if $\Delta$ is not a cone and $\card{\facets I}=2$ then $\euler\Delta=1$.
\label{thm:simBasecase:two}
\end{enumerate}
\end{theorem}
\begin{proof}
\proofPart{\eqref{thm:simBasecase:support}}
This is well known and not hard to prove.

\proofPart{\eqref{thm:simBasecase:prime}} By induction on
$\card{\facets\Delta}$ using Lemma \ref{lem:codisjointFacet}.

\proofPart{\eqref{thm:simBasecase:two}} The two facets are co-disjoint
when ignoring unused vertices.
\end{proof}

\begin{lemma}
\label{lem:codisjointFacet}
Let $D$ be a simplicial complex and let $\sigma\in\facets D$ such that
$\sigma\neq\vertices D$ is co-disjoint to every other facet of
$D$. Then $\euler D=-\euler{\Delta\vsub\comple\sigma}$ where
$\Delta\defeq\clo{\facets D\setminus\set\sigma}$.
\end{lemma}
\begin{proof}
Let $v\notin\sigma$ and $\tau\in\facets\Delta$. Then $v\in\tau$ so
$\Delta$ is a cone and $\euler\Delta=0$. By Theorem \ref{thm:com}
\[
\euler D =
\euler{\Delta\cup\pows p} =
\euler\Delta - \euler{\Delta\vsub\comple\sigma} =
-\euler{\Delta\vsub\comple\sigma}.
\qedhere
\]
\end{proof}

\subsection{Termination and Complexity}

It is clear that the simplicial DBMS algorithm terminates since
$\card{\facets\Delta}$ decreases strictly at each step. For the
simplicial BCRT algorithm, termination requires that we choose the
pivots $\sigma$ such that $\sigma\subsetneq\vertices\Delta$ and
$\sigma\notin\Delta$ since otherwise we get an infinite number of
steps from $\Delta=\Delta\vsub\comple\sigma$ or
$\Delta=\Delta\cup\pows\sigma$

If we cannot choose a $\sigma$ such that
$\sigma\subsetneq\vertices\Delta$ and $\sigma\notin\Delta$ then
$\Delta=\pows{\vertices\Delta}$ which is a base case. If we always
choose pivots $\sigma$ such that $\sigma\subsetneq\vertices\Delta$ and
$\sigma\notin\Delta$ then termination of the simplicial BCRT algorithm
follows from the fact that either the number of vertices
$\card{\vertices\Delta}$ or the number of non-faces
$\card{\comple\Delta}=2^{\card{\vertices\Delta}}-\card\Delta$
decreases between any two steps. In fact some consideration shows that
$\card{\comple\Delta}$ decreases strictly at each step.

We have seen that the DBMS algorithm gets rid of at least one facet at
each step. The BCRT algorithm gets rid of at least one vertex at each
step if the pivot is chosen to be of the form
$\sigma\defeq\comple{\set e}$. It is immediate that $e$ is not a
vertex of $\Delta\vsub\comple\sigma$. To see that $e$ can also be
removed from $\Delta\cup\pows\sigma$, observe that $\sigma$ is the
only facet that does not contain $e$, so $e$ is
$(\Delta\cup\pows\sigma)$-independent from the other vertices and so
can be removed using the independent vertices technique from Section
\ref{sec:simImp}.

Since the Euler characteristic problem is \sPc{} we expect both
algorithms to run in at least single exponential time. Taking nerves
as described in Section \ref{sec:simImp}, we can interchange the
number of vertices $\card{\vertices\Delta}$ with the number of facets
$\card{\facets\Delta}$. So if
$l\defeq\min(\card{\vertices\Delta},\card{\facets\Delta})$ then both
the BCRT and DBMS algorithms have $O(q2^l)$ asymptotic worst case time
complexity where $q$ is a polynomial. So both algorithms run in single
exponential time. We expect that a more careful analysis could reduce
the base of the exponential.

\subsection{Pivot Selection}
\label{sec:simPivots}

We have proven that the BCRT and DBMS algorithms terminate in a finite
amount of time. To be useful in practice, the amount of time until
termination should be small rather than just finite. The strategy used
for selecting pivots when splitting an ideal has a significant impact
on performance. We describe several different pivot selection
strategies here and compare them empirically in Section
\ref{sec:bench}.

A \emph{rare vertex} is a vertex that belongs to a minimum number of
facets. A \emph{popular vertex} is a vertex $e$ that is an element of
a maximum number of facets with the constraint that $\comple{\set
  e}\notin\facets\Delta$.

If there are several candidate pivots that fit a given pivot selection
strategy, then the pivot used is chosen in an arbitrary deterministic
way among the tied candidates.

\subsubsection*{BCRT Pivot Selection Strategies}
Recall that BCRT pivots $\sigma$ satisfy
$\sigma\subsetneq\vertices\Delta$ and $\sigma\notin\Delta$.

The reason that some of these names seem opposite of their definition
is that the names come from the algebraic setting and the translation
to simplicial complexes involves taking a complement. For example a
rare variable of an ideal translates to a popular vertex of the
corresponding simplicial complex.

\begin{description}
\item[popvar] The pivot is $\comple{\set e}$ for $e$ a rare vertex.
\item[rarevar] The pivot is $\comple{\set e}$ for $e$ a popular vertex.
\item[random] The pivot is $\comple{\set e}$ for $e$ a random vertex
  such that $\comple{\set e}\notin\facets\Delta$.
\item[popgcd] Let $e$ be a rare vertex. The pivot is the union of
  three facets chosen uniformly at random among those facets that do
  not contain $e$.
\end{description}

The simplicial BCRT algorithm presented here is based on the algebraic
BCRT algorithm for \hps{}, and the strategies \pname{popvar} and
\pname{popgcd} have been found to work well for \hps{} computation. So
we also try them here. \pname{rarevar} and \pname{random} have been
included to have something to compare to.

\subsubsection*{DBMS Pivot Selection Strategies}

Recall that DBMS pivots are elements of $\facets\Delta$.

\begin{description}
\item[rarevar] The pivot is a facet that does not contain some popular
  vertex.
\item[popvar] The pivot is a facet that does not contain some rare vertex.
\item[maxsupp] The pivot is a facet of minimum size.
\item[minsupp] The pivot is a facet of maximum size.
\item[random] The pivot is a facet chosen uniformly at random.
\item[rarest] The pivot is a facet that lacks (does not contain) a
  maximum number of popular vertices. Ties are broken by picking the
  facet that lacks the maximum number of second-most-popular vertices
  and so on.
\item[raremax] The pivot is chosen according to \pname{rarevar} where
  ties are broken according to \pname{maxsupp}.
\end{description}

\subsection{Improvements}
\label{sec:simImp}

We present several improvements to the DBMS and BCRT algorithms.

\subsubsection*{Independent Vertices}

If $\Delta$ and $\Gamma$ are simplicial complexes with disjoint vertex
sets $\vertices\Delta$ and $\vertices\Gamma$, then
\[
\Delta\simProd\Gamma\defeq
(\Delta\times\pows{\vertices\Gamma}) \cup
(\pows{\vertices\Delta}\times\Gamma).
\]

So if $\Delta\defeq\clo{\set x,\set y}$ and $\Gamma\defeq\clo{\set a,\set b,\set c}$
then
\[
\Delta\simProd\Gamma=\clo{
  \set{x,a,b,c},
  \set{y,a,b,c},
  \set{x,y,a},
  \set{x,y,b},
  \set{x,y,c}}.
\]
Using Theorem \ref{thm:simProd} we can compute
$\euler{\Delta\simProd\Gamma}$ in terms of $\euler\Delta$ and
$\euler\Gamma$. So if we are computing $\euler\Psi$ for a simplicial
complex $\Psi$, then we could significantly simplify the task by
finding simplicial complexes $\Delta$ and $\Gamma$ such that
$\Psi=\Delta\simProd\Gamma$.

We say that two subsets $A,B\subseteq\vertices\Psi$ are
\emph{$\Psi$-independent} if $\facets\Psi$ is the disjoint union
of
\[
F_A\defeq
\setBuilder{\sigma\in\facets\Psi}{\sigma\setminus A=\comple A}
\ \text{ and }\ 
F_B\defeq
\setBuilder{\sigma\in\facets\Psi}{\sigma\setminus B=\comple B}.
\]
If $\Psi=\Delta\simProd\Gamma$ then it is immediate that
$\vertices\Delta$ and $\vertices\Gamma$ are $\Psi$-independent. On the
other hand, if $A$ and $B$ are $\Psi$-independent, then
$\Psi=\Delta\simProd\Gamma$ where
\[
\Delta\defeq \clo{F_A}\vsub\comple A
\ \text{ and }\ 
\Gamma\defeq \clo{F_B}\vsub\comple B.
\]
We have proven Theorem \ref{thm:simDecom}, which together with Theorem
\ref{thm:simProd} generalizes Lemma \ref{lem:codisjointFacet}.
\begin{theorem}
\label{thm:simDecom}
A simplicial complex $\Psi$ can be written as
$\Psi=\Delta\simProd\Gamma$ if and only if there is a
$\Psi$-independent pair $(A,B)$.
\end{theorem}

\begin{theorem}
\label{thm:simProd}
Let $\Delta$ and $\Gamma$ be simplicial complexes with disjoint
non-empty vertex sets $\vertices\Delta$ and $\vertices\Gamma$. Then
\[
\euler{\Delta\simProd\Gamma}=\euler\Delta\euler\Gamma.
\]
\end{theorem}
\begin{proof}
Choose any set $D$ such that $\Delta=\clo D$. Then we get by inclusion-exclusion that
\begin{align*}
\euler\Delta =
\euler{\bigcup_{\sigma\in D}\pows\sigma} &=
\sum_{v\subseteq D}(-1)^{\card v+1}\euler{\bigcap_{\sigma\in v}\pows\sigma} \\&=
\sum_{v\subseteq D}(-1)^{\card v+1}\euler{\pows{\cap v}}.
\end{align*}
Let $F_\Delta\defeq\facets\Delta\times\set{\vertices\Gamma}$ and
$F_\Gamma\defeq\set{\vertices\Delta}\times\facets\Gamma$. Then
$\Delta\simProd\Gamma=\clo{F_\Delta\cup F_\Gamma}$ so in the same way
we get by inclusion-exclusion that
\begin{align*}
\euler{\Delta\simProd\Gamma} &=
\sum_{v\subseteq F_\Delta\cup F_\Gamma}(-1)^{\card v+1}\euler{\pows{\cap v}} \\&=
\sum_{d\subseteq F_\Delta}\sum_{g\subseteq F_\Gamma}
  (-1)^{\card d+\card g+1} \euler{\pows{\cap(d\cup g)}}.
\end{align*}
Let $d\p\subseteq\facets\Delta$ and $g\p\subseteq\facets\Gamma$ such
that $d=d\p\cup\vertices\Gamma$ and
$g=g\p\cup\vertices\Delta$. We will prove that
\begin{equation}
\label{eqn:powsDecom}
\euler{\pows{\cap(d\cup g)}} =
-\euler{\pows{\cap d\p}} \euler{\pows{\cap g\p}}.
\end{equation}
Since $\card d=\card{d\p}$ and $\card g=\card{g\p}$ we then get that
\begin{align*}
\euler{\Delta\simProd\Gamma} &=
\sum_{d\p\subseteq\facets\Delta}\sum_{g\p\subseteq\facets\Gamma}
  (-1)^{\card{d\p}+\card{g\p}}
  \euler{\pows{\cap d\p}}\euler{\pows{\cap g\p}} \\&=
\sum_{d\p\subseteq\facets\Delta}(-1)^{\card{d\p}+1}\euler{\pows{\cap d\p}}
  \sum_{g\p\subseteq\facets\Gamma}(-1)^{\card{g\p}+1}\euler{\pows{\cap g\p}} \\&=
\euler\Delta\euler\Gamma.
\end{align*}

It only remains to prove Equation \eqref{eqn:powsDecom}. Observe that
$\euler{\pows\tau}=0$ if $\tau$ is non-empty and otherwise
$\euler{\pows{\tau}}=-1$. So Equation \eqref{eqn:powsDecom} is
equivalent to the statement that
\[
\cap(d\cup g)=\emptyset \quad\biimp\quad
\cap d\p=\emptyset\text{ and }\cap g\p=\emptyset.
\]
As every face in $d$ contains $\vertices\Gamma$ and every face in $g$
contains $\vertices\Delta$, the only way for $\cap(d\cup g)$ to be
empty is if $\cap d=\vertices\Gamma$ and $\cap g=\vertices\Delta$. The
statement then follows from the observation that $\cap
d=\vertices\Gamma$ if and only if $\cap d\p=\emptyset$ and likewise
$\cap g=\vertices\Delta$ if and only if $\cap g\p=\emptyset$.
\end{proof}

\subsubsection*{Eliminate Abundant Vertices}
Suppose $e$ is an element of every facet of $D$ except
$\sigma\in\facets D$. Use Equation \eqref{eqn:simSplitFacet} with
$\sigma$ as the pivot to get that
\[
\euler D =
\euler{\Delta\cup\pows\sigma} =
\euler\Delta - \euler{\Delta\vsub\comple\sigma} =
\euler{\Delta\vsub\comple\sigma},
\]
since $\euler\Delta=0$ as every facet of $\Delta$ contains $e$ so
$\Delta$ is a cone.

\subsubsection*{Take Nerves}

The \emph{nerve} of a simplicial complex $\Delta$ is
\[
\nerve\Delta \defeq
\setBuilder{v\subseteq\facets\Delta}{\cap v\neq\emptyset}.
\]
A complex and its nerve have the same homotopy type
\cite[Thm. 10]{nerves}, so $\euler\Delta=\euler{\nerve\Delta}$. Taking
the nerve corresponds to transposing the facet-vertex incidence matrix
$M$. Taking the nerve twice will remove dominated columns from $M$ and
then remove any dominated rows that might have appeared. This process
can continue as removing rows of $M$ can cause yet more columns to be
dominated.

If no columns are dominated then taking the nerve once will exchange
the number of vertices and facets. The BCRT algorithm is more
sensitive to the number of vertices than it is to the number of
facets, so it can be beneficial to do computations on the nerve if the
complex has fewer facets than vertices. The DBMS algorithm is opposite
of this in that it is more sensitive to the number of facets than the
number of vertices.

\subsubsection*{Base Case for $\card{\facets\Delta}=3$}
Assume that all abundant vertices have been removed, that $\Delta$ is
not a cone and that $\card{\facets\Delta}=3$. Then every vertex that
actually occurs in the complex is an element of one facet. The nerve
of the nerve of $\Delta$ will then have the form $\clo{\set a, \set b,
  \set c}$. So $\euler\Delta=2$.

\subsubsection*{Partial Base Case for $\card{\facets\Delta}=4$}
Suppose that $\card{\facets\Delta}=4$, that every vertex is an element
of exactly two facets and that the number of vertices is 4. Then
$\euler\Delta=-1$. There should be more rules like this, though
identifying them by hand is laborious and error prone.

\subsubsection*{Make a Table}
It would be beneficial for each small $k$ to perform a computer search
to make a table of all simplicial complexes $\Delta$ with
$\card{\facets\Delta}=k$ up to reordering of the vertices and the
various techniques for simplifying a simplicial complex that we have
presented. Then the Euler characteristic of simplicial complexes with
few facets could be computed through a table look-up.

\subsubsection*{Data Structures}
In our implementation we have represented a facet as a 0-1 vector,
where we pack 32 or 64 entries into a single 32 or 64 bit word. Many
operations are very fast in this representation. While the general
concept is simple we warn that the implementation details are tricky.

For sparse or complement-of-sparse vectors, it might pay off to only
record the zero entries or one entries respectively, though this is
not something that we have pursued.

\section{Ideals to Simplicial Complexes}
\label{sec:translation}

In this section we describe in more detail how the simplicial
algorithms in Section \ref{sec:simAlg} are equivalent to the algebraic
algorithms in Section \ref{sec:algAlg}. We treat this topic in detail
in part to support our claim that the algebraic and simplicial
algorithms are equivalent and in part because the algebra-simplicial
translation brings up interesting mathematics.

To set up the algebra-simplicial translation, let $\Delta$ be a
simplicial complex with vertices $x_1,\ldots,x_n$ that are
simultaneously the variables in a polynomial ring. The
algebra-simplicial connection is via the function $\comtoidSym$ from
sets to monomials defined by
$\comtoid\sigma\defeq\Pi\comple\sigma$. In other words,
$\comtoid\sigma$ is the product of variables (vertices) that are not
in $\sigma$. So if $n=3$ then $\comtoid\emptyset=x_1x_2x_3$ and
$\comtoid{\set{x_1,x_3}}=x_2$.

Extend $\comtoidSym$ from sets to simplicial complexes by
$\comtoid\Delta\defeq\idealBuilder{\comtoid\sigma}{\sigma\in\Delta}.$
Then $\comtoidSym$ is a bijection from the facets of $\Delta$ to the
minimal generators of $\comtoid\Delta$, so $\comtoid\Delta$ can be
quickly computed given $\Delta$. We can also describe $\comtoid\Delta$
as the \emph{Stanley-Reisner ideal} of the \emph{Alexander dual} of
$\Delta$, though we will not use this alternative description of
$\comtoid\Delta$ here.

The fundamental observation that we have been using is that
$\euler{\comtoid\Delta}=\euler\Delta$ so that we can compute
$\euler\Delta$ with monomial ideal methods on $\comtoid\Delta$.

\begin{theorem}
If $\Delta$ is a simplicial complex then
$\euler\Delta=\euler{\comtoid\Delta}$, where $\euler{\comtoid\Delta}$ is the
coefficient of $\varProd$ in the multivariate \hps{} numerator
$\hilbNum I$ of $I$.
\end{theorem}
\begin{proof}
The proof is based on a formula of Bayer concerning upper Koszul
simplicial complexes.

Let $m$ be a monomial and let $I$ be a monomial ideal. Then the
\emph{upper Koszul simplicial complex}\footnote{This notion of Koszul
  complex is not to be confused with the chain complex notion of a
  Koszul complex.} \cite{cca,bayerMonomialLectureNotes} of $I$ at $m$
is
\[
\Delta^I_m\defeq
\setBuilder{\sigma\subseteq\set{x_1,\ldots,x_n}}{\frac{x^m}{\Pi\sigma}\in
  I}.
\]
The numerator of the multivariate \hps{} of $I$ is related to the
Euler characteristics of the Koszul complexes of $I$ by the following
formula due to Bayer \cite[Proposition 3.2]{bayerMonomialLectureNotes}
\[
\hilbNum I - 1 =
\sum_{v\in\N^n}\euler{\Delta^I_ {x^v}}x^v.
\]
Since $\koz {\comtoid\Delta} \varProd=\Delta$, we get that
\[
\textrm{Coefficient of $\varProd$ in $\hilbNum{\comtoid\Delta}$} =
\euler {\koz {\comtoid\Delta} \varProd} =
  \euler \Delta.
\qedhere
\]
\end{proof}

\subsection{The Inverse of $\comtoidSym$}
\label{sec:inverseComtoid}

In Section \ref{sec:simBackground} we modified the definition of a
simplicial complex to have an associated set of vertices so that a
vertex need not actually appear in any of the faces of the
complex. Here we give an example that shows that this change is
necessary for $\phi$ to be a bijection between simplicial complexes
and square free monomial ideals. In short, the vertex set definition
is necessary to preserve information about the variables in the
ambient polynomial ring.

We have two functions $\comtoidSym$, one that maps sets to square free
monomials and one that maps ideals to square free monomial
ideals. They are defined by $\comtoid v\defeq\Pi v$ and
$\comtoid\Delta\defeq\idealBuilder{\comtoid v}{v\in\Delta}$
respectively. The inverse of $\Delta\mapsto\comtoid\Delta$ is given by
\[
\idtocom I=\setBuilder{v\subseteq\set{x_1,\ldots,x_n}}{\comtoid v\in I}.
\]
We define the vertex set of $\idtocom I$ to be the set of variables in
the ambient ring of $I$, even if some of those variables do not appear
in any face of $\idtocom I$. We give an example that shows that
$\Delta\mapsto\comtoid\Delta$ would not have an inverse if the vertex
set of $\idtocom I$ were the union of its faces.

Consider the ideal
$I\defeq\ideal{x_1x_2,x_1x_3}\subseteq\kappa[x_1,x_2,x_3]$ and let
$\Delta\defeq\idtocom I$ so that
$\facets\Delta=\set{\set{x_3},\set{x_2}}$. The question now is what
the vertex set of $\Delta$ is. The union of faces is $\set{x_2,x_3}$
so if that were the set of vertices then $\comtoid{\set{x_2}}=x_3$ and
$\comtoid{\set{x_3}}=x_2$ so then $\comtoid{\idtocom
  I}=\comtoid\Delta=\ideal{x_3,x_2}\neq I$. We observe that if the
vertex set were the union of faces, then there would be no simplicial
complex $\Delta$ such that $\comtoid\Delta=I$. So $\comtoidSym$ would
be a bijection not from simplicial complexes to square free monomial
ideals, but from simplicial complexes to square free monomial ideals
with the property that $\gcd(\ming I)=1$.

Using the definition that we have given, the vertex set of $\Delta$ is
still $\set{x_1,x_2,x_3}$ even though $x_1$ does not appear in any
face of $\Delta$. With our definition $\comtoidSym$ does have
$\idtocomSym$ as an inverse, and indeed we see that $\comtoid{\idtocom
  I}=\comtoid\Delta=I$.

Even if the input ideal $I$ to the algebraic BCRT and DBMS algorithms
has the property that $\gcd(\ming I)=1$, the intermediate ideals that
these algorithms generate do not necessarily have that property. So if
we had defined the vertex set to be the union of faces, then we would
have had to deal with this issue in some other way.

\subsection{Divide...}

The algebra algorithms are based on Equation \ref{eqn:alg} which
states that for monomials $p$
\[
\euler I = \euler{I:p} + \euler{I + \ideal p}.
\]
The simplicial algorithms are based on Theorem \ref{thm:com} which
states that for sets $\sigma$
\[
\euler\Delta =
 \euler{\Delta\vsub\comple\sigma} +
 \euler{\Delta\cup\pows\sigma}.
\]
These two equations are equivalent since if $\Delta=\idtocom I$ and
$\sigma=\idtocom p$ then
\[
\idtocom{I:p}=\Delta\vsub\comple\sigma,\quad\quad
\idtocom{I+\ideal p}=\Delta\cup\pows\sigma.
\]

\subsection{... and Conquer}

Theorem \ref{thm:basecase} specifies base cases for the algebra algorithms when
\begin{enumerate}
\item $I$ does not have full support,
\item $I$ has full support and the minimal generators $\ming I$ are
  pairwise prime,
\item $I$ has full support and $\card{\ming I}=2$.
\end{enumerate}
Theorem \ref{thm:simBasecase} specifies base cases for the simplicial
algorithms when
\begin{enumerate}
\item $\Delta$ is a cone,
\item $\Delta$ is not a cone and the facets of $\Delta$ are
  pairwise co-disjoint,
\item $\Delta$ is not a cone and $\card{\facets I}=2$.
\end{enumerate}
If $\Delta=\idtocom I$ then the algebra and simplicial conditions are
equivalent. To be more precise
\begin{enumerate}
\item $I$ has full support if and only if $\Delta$ is not a cone,
\item monomials $a$ and $b$ are relatively prime if and only if
  $\idtocom a$ and $\idtocom b$ are co-disjoint,
\item $\card{\ming I}=\card{\facets\Delta}$.
\end{enumerate}

\subsection{Termination and Complexity}

The algebraic DBMS algorithm terminates as $\card{\ming I}$ decreases
strictly at each step. The simplicial DBMS algorithm terminates as
$\card{\facets\Delta}$ decreases strictly at each step. These reasons
are equivalent when $\Delta=\idtocom I$ as then $\card{\ming
  I}=\card{\facets\Delta}$.

The algebraic BCRT algorithm terminates if the pivots $p$ are chosen
such that $p\neq 1$ and $p\notin I$. The simplicial BCRT algorithm
terminates if the pivots $\sigma$ are chosen such that
$\sigma\subsetneq\vertices\Delta$ and $\sigma\notin\Delta$. If
$\Delta=\idtocom I$ and $\sigma=\idtocom p$ then these conditions are
equivalent.

\subsection{Pivot Selection}

The pivot selection strategies that have the same name in the two
sections on pivot selection are equivalent. The main points in proving
this are that for a monomial ideal $I$ and a simplicial complex
$\Delta\defeq\idtocom I$
\begin{itemize}
\item a variable $x_i$ is rare if and only if the vertex $x_i$ is popular,
\item a variable $x_i$ is popular if and only if the vertex $x_i$ is rare,
\item a minimal generator $m$ has maximum support if and only if $\idtocom m$
  has minimum size,
\item a minimal generator $m$ has minimum support if and only if $\idtocom m$
  has maximum size.
\end{itemize}

\subsection{Improvements}

Let $I$ be a square free monomial ideal and let $\Delta\defeq\idtocom
I$. The two sections on improvements present equivalent techniques in
the same order.

\subsubsection*{Independent Variables --- Independent Vertices}
Let $A$ and $B$ be disjoint sets of variables. Let
$I\subseteq\kappa[A]$ and $\quad J\subseteq\kappa[B]$. Then
$I+J\subseteq\kappa[A\cup B]$ and the algebraic section on independent
variables proves that $\euler{I+J}=\euler I\euler J$. Furthermore, an
ideal $K$ can be written as a sum $I+J$ if and only if there are
disjoint sets $A$ and $B$ such that $\ming K$ is the union of $\ming
K\cap\kappa[A]$ and $\ming K\cap\kappa[B]$. In that case we say that
$A$ and $B$ are $K$-independent and then $K=I+J$ for
$I\defeq\ideal{\ming K\cap\kappa[A]}$ and $J\defeq\ideal{\ming
  K\cap\kappa[B]}$.

For the simplicial side of things, let $\Delta$ and $\Gamma$ be
simplicial complexes such that the vertex sets $\vertices\Delta$ and
$\vertices\Gamma$ are disjoint. Then 
\[
\Delta\simProd\Gamma\defeq
(\Delta\times\pows{\vertices\Gamma}) \cup
(\pows{\vertices\Delta}\times\Gamma)
\quad\text{and}\quad
\vertices{\Delta\simProd\Gamma}\defeq
\vertices\Delta\cup\vertices\Gamma
\]
The simplicial section on independent vertices proves that
$\euler{\Delta\simProd\Gamma}=\euler\Delta\euler\Gamma$. Furthermore,
an ideal $\Psi$ can be written as $\Delta\simProd\Gamma$ if and only
if there are disjoint sets $A$ and $B$ such that $\facets\Psi$ is the
disjoint union of
\[
F_A\defeq
\setBuilder{\sigma\in\facets\Psi}{\sigma\setminus A=\comple A}
\ \text{ and }\ 
F_B\defeq
\setBuilder{\sigma\in\facets\Psi}{\sigma\setminus B=\comple B}.
\]
In that case $\Psi=\Delta\simProd\Gamma$ for
$\Delta\defeq\clo{F_A}$ and $\Gamma\defeq\clo{F_B}$.

Addition of ideals and $\simProd$ of simplicial complexes are
equivalent. If $\Delta\defeq\idtocom I$ and $\Gamma\defeq\idtocom J$
then $\Delta\simProd\Gamma=\idtocom{I+J}$. Observe that we are using
three different functions $\idtocomSym$ here as the ambient polynomial
ring is different for each of them. Recall that the definition of
$\idtocomSym$ involves taking a complement, and if $v$ is a set of
variables then the meaning of the complement $\comple v$ depends on
what the variables in the ambient polynomial ring are.

Independence of variables and vertices are also equivalent. If
$\Psi\defeq\idtocom K$ and $A$ and $B$ are sets of variables/vertices,
then $A$ and $B$ are $K$-independent if and only if $A$ and $B$ are
$\Psi$-independent.

Let $A$ and $B$ be $K$-independent and let $\Psi=\idtocom K$ so that
$A$ and $B$ are also $\Psi$-independent. Let $\Delta\defeq\ideal{F_A}$
and $\Gamma\defeq\ideal{F_B}$. Let $I\defeq\ideal{\ming
  K\cap\kappa[A]}$ and $J\defeq\ideal{\ming K\cap\kappa[B]}$. Then
$K=I+J$ and $\Psi=\Delta\simProd\Gamma$. Furthermore,
$\Delta=\idtocom I$ and $\Gamma=\idtocom J$.

The notion of $I$-independence is standard, though we have not been
able to find any reference in the literature to $\Psi$-independence or
the operation $\simProd$ on simplicial complexes.

\subsubsection*{Eliminate Unique Variables --- Eliminate Abundant Vertices}
A variable $x_i$ is unique in $I$ if and only if $x_i$ is an abundant
vertex of $\Delta$.

\subsubsection*{Transpose Ideals --- Take Nerves}
Both the nerve and the transpose of an ideal are transposing a matrix
where the columns are variables/vertices and the rows are
generators/facets. One of these matrices can be derived from the other
by replacing all 0's by 1 and simultaneously replacing all 1's by
0. This shows that $\idtocom{\idtrans I} = \nerve\Delta$.

As far as we know, the transpose operation has not been applied to
monomial ideals before -- we are investigating if a monomial ideal and
its transpose have any interesting relations between them when the
ideal is not square free.

\subsubsection*{Base Case for $\card{\ming I}=3$ --- Base Case for $\card{\facets\Delta}=3$}
These are equivalent as $\card{\ming I}=\card{\facets\Delta}$ and both
base cases give the same Euler characteristic of 2.

\subsubsection*{Partial Base Case for $\card{\ming I}=4$ --- Partial Base Case for $\card{\facets\Delta}=4$} These are equivalent.

\subsubsection*{Make a Table} This is the same idea.

\subsubsection*{Data Structures}
These are the same considerations. Sparse exponent vectors correspond
to large facets, while small facets correspond to mostly-one exponent
vectors. In particular, it is not the case that sparse exponent
vectors correspond to small facets.

\section{Empirical Evaluation of Euler Characteristic Algorithms}
\label{sec:bench}

We have implemented the BCRT and DBMS algorithms for Euler
characteristic in the program \Frobby{} \cite{frobby}. In this
section we explore the pivot selection strategies for both algorithms,
and we compare the BCRT and DBMS implementations to several other
systems with Euler characteristic functionality.

We use simplicial terminology in this section. To recover equivalent
statements using algebraic terminology read ``monomial ideal'' for
``complex'', read ``variables'' for ``vertices'', read ``minimal
generators'' for ``facets'', read ``having not full support'' for
``being a cone'' and read ``transpose'' for ``nerve''.

We compare the following implementations, listed in alphabetical order.

\begin{description}
\item[\Frobby{} version 0.9.0 \cite{frobby}] \Frobby{} is a free and
  open source system for computations on monomial ideals. We have
  written a C++ implementation of the BCRT and DBMS algorithms for
  Euler characteristic in \Frobby{}. We ran Frobby with the option
  {\tt -minimal} turned on in Table \ref{tab:allPrograms} as all faces
  given in the input are facets.

\item[\GAP{} version 4.4.12 \cite{gap}] \GAP{} is a free and open
  source system for computational discrete algebra. It computes the
  Euler characteristic of a complex by enumerating all faces. The
  implementation is written in the \GAP{} scripting language. We use
  the {\tt SCEulerCharacteristic} function from version 1.4.0 of the
  {\tt simpcomp} package \cite{simpcomp}. We extend the memory limit
  for \GAP{} to 2 GB.

\item[\Mtwo{} version 1.4 \cite{m2}] \Mtwo{} is a free and open
  source computer algebra system. It computes the Euler characteristic
  of a simplicial complex using the BCRT algorithm for \hps{} to get
  the $f$-vector. The time consuming parts of the code are written in
  C++. Due to some inefficiencies in the {\tt SimplicialComplexes}
  package it is faster to compute the Euler characteristic using the
  {\tt poincare} function directly on a monomial ideal instead of
  going through a simplicial complex. We have used this faster method.

\item[\Sage{} version 4.7 \cite{sage}] \Sage{} is a free and open
  source computer algebra system. It computes the Euler characteristic
  of a simplicial complex by enumerating all faces of the complex. The
  implementation is written in Python. We use the {\tt
    eulerCharacteristic} function from the {\tt SimplicialComplexes}
  package.
\end{description}

We use the following complexes for the comparison.

\begin{description}
\item[random-$a$-$b$] A randomly generated complex with $a$ vertices
  and $b$ facets. The complex is generated one facet at a time. A
  prospective facet $\sigma$ is generated at random so that each
  vertex has a 50\% chance to be an element of $\sigma$. If $\sigma$
  is contained in or contains any previously generated facet then
  $\sigma$ is discarded and another prospective facet is generated in
  its stead.

  We generated these example using the {\tt genideal} action of
  \Frobby{}.

\item[nicgraph-$a$-$b$] The simplicial complex of all not
  $b$-connected graphs on a given set of $a$ vertices. Here we view
  each possible edge as a vertex of the simplicial complex. A graph is
  $b$-connected if it is connected and it cannot be disconnected by
  removing $b-1$ edges.

  We genenerated these examples with the {\tt NotIConnectedGraphs}
  function in \Sage{}.

\item[rook-$a$-$b$] The $a\times b$ \emph{chessboard complex}.

  Let $V$ and $W$ be sets of vertices such that $\card V=a$ and $\card
  W=b$. A \emph{partial matching} between $V$ and $W$ is a set of
  pairs $(v,w)$ with $v\in V$ and $w\in W$ such that each vertex is in
  at most one pair. The set of partial matchings between $V$ and $W$
  forms a simplicial complex called the \emph{chessboard complex}. It
  corresponds to ways of placing rooks on an $a\times b$ chessboard so
  that no rook attacks any other rook.

  We initially generated these examples with the {\tt
    ChessboardComplex} function in \Sage{}. We could not generate
  large enough examples with that function so we added the same
  functionality to \Frobby{} and used that.

\item[match-$a$] The matching complex on $a$ vertices.

  Let $V$ be a set of vertices with $\card v=a$. A \emph{partial
    matching} on $V$ is a set of pairs $\set{v,w}$ for $v,w\in V$ such
  that each vertex is in at most one pair. The set of partial
  matchings on $V$ forms a simplicial complex called the
  \emph{matching complex}.

  We initially generated these examples with the {\tt MatchingComplex}
  function in \Sage{}. We could not generate large enough examples
  with that function so we added the same functionality to \Frobby{}
  and used that.
\end{description}

In all tables $\card{\vertices\Delta}$ refers to the number of
vertices and $\card{\facets\Delta}$ refers to the number of
facets. OOM stands for ``out of memory'' and indicates that the
computation was terminated due to the program reporting an out of
memory error. All experiments were run on an Intel\textregistered{}
Core\texttrademark{} 2 Duo CPU T7500 at 2.20GHz with 4GB of RAM
running Mandriva Linux 2010.1.

We checked that all the programs gave the correct answer for every
input. All times are the median time of three runs. Time is in each
case measured by the programs themselves, except for \Frobby{} that was
timed using the {\tt Unix} time command line utility. All times
exclude the time taken to start the program and read the input file,
except for the \Frobby{} times that do include startup and input time. We
had to cut out startup and input time as it was taking more time than
the Euler characteristic computation itself in some cases, and these
experiments are not intended to be about starting programs or reading
input. We did not do this for \Frobby{} as \Frobby{} starts up and reads
input in a tiny fraction of the time taken for the Euler
characteristic computation.

\subsection{Pivot Selection Strategies for the BCRT and DBMS Algorithms}

The BCRT and DBMS algorithms are parameterized on a pivot selection
strategy that determines which pivot to use at each step. We compare
the strategies described in Section \ref{sec:algPivots}/Section
\ref{sec:simPivots} on a battery of randomly generated complexes.

\subsubsection*{BCRT Pivots}

Table \ref{tab:bcrtPivotNerve} shows that \pname{popvar} is
the best BCRT pivot selection strategy for these randomly generated
complexes. \pname{popvar} is already faster for the simpler ideals and
its lead over the other strategies increases with the number of facets
and especially with the number of vertices.

\begin{description}
\item[random] This strategy is here to be able to tell if the other
  strategies are better or worse than a random choice of vertex.
\item[popvar] We believe that \pname{popvar} performs well because a
  rare vertex $e$ gives $\Delta\cup\ideal{\comple{\set e}}$ as few facets as
  possible.
\item[rarevar] \pname{rarevar} does the opposite of what
  \pname{popvar} does. It is never far ahead and is mostly
  significantly behind \pname{popvar}. For the balanced examples it is
  much worse than even a random choice of vertex.
\item[popgcd] \pname{popgcd} works well for the \hps{} algorithms that
  the algorithms we present here are based on. However, the data shows
  that it is often significantly worse than choosing a random vertex
  when it comes to computing Euler characteristics.
\end{description}

\begin{table}
\centering
\begin{tabular}{rrrrrrrr}
$\card{V_\Delta}$ & $\card{\facets\Delta}$ & \tt random & \tt popvar & \tt rarevar & \tt popgcd \\
\hline
4000 & 20 & {\bf 0.05} & {\bf 0.05} & 0.06 & 0.06\\
6000 & 20 & {\bf 0.07} & {\bf 0.07} & {\bf 0.07} & {\bf 0.07} \\
8000 & 20 & 0.09 & 0.09 & {\bf 0.08} & 0.09 \\
10000 & 20& {\bf 0.11} & {\bf 0.11} & {\bf 0.11} & 0.12 \\
\hline
4000 & 30 & 3.70 & {\bf 2.31} & 5.73 & 9.83 \\
6000 & 30 & 4.84 & {\bf 3.25} & 6.31 & 13.41\\
8000 & 30 & 5.87 & {\bf 4.18} & 7.61 & 15.56 \\
10000 & 30& 7.25 & {\bf 5.34} & 9.59 & 20.53\\
\hline
20 & 4000 & 0.13 & {\bf 0.12} & {\bf 0.12} & 0.24 \\
20 & 6000 & 0.24 & 0.24 & {\bf 0.22} & 0.45 \\
20 & 8000 & 0.37 & 0.37 & {\bf 0.36} & 0.65 \\
20 & 10000 & 0.55 & 0.55 & {\bf 0.52} & 0.95\\
\hline
30 & 4000 & 4.11 & {\bf 2.63} & 5.78 & 14.53\\
30 & 6000 & 6.04 & {\bf 4.06} & 7.69 & 22.01\\
30 & 8000 & 7.48 & {\bf 5.37} & 9.11 & 28.85\\
30 & 10000 & 9.48 & {\bf 6.99} & 10.83 & 33.41\\
\hline
100 & 100 & 7.07 & {\bf 1.38} & 40.15 & 166.56\\
120 & 120 & 32.01 & {\bf 5.92} & 170.28 & 2109.31\\
140 & 140 & 130.44 & {\bf 23.19} & 740.48 & $>$7200\\
160 & 160 & 491.43 & {\bf 89.87} & 2729.22 & $>$7200\\
180 & 180 & 1497.50 & {\bf 261.09} & $>$7200 & $>$7200\\
200 & 200 & 4965.33 & {\bf 796.03} & $>$7200 & $>$7200\\
220 & 220 & $>$7200 & {\bf 1756.16} & $>$7200 & $>$7200\\
240 & 240 & $>$7200 & {\bf 5051.55} & $>$7200 & $>$7200\\
\smallskip
\end{tabular}
\caption{BCRT pivot selection strategies. All times are in seconds.}
\label{tab:bcrtPivotNerve}
\end{table}

\subsubsection*{DBMS Pivots}

Table \ref{tab:dbmsPivotNerve} shows that \pname{raremax} is the
best DBMS pivot selection strategy for these randomly generated
complexes.

\begin{description}
\item[random] This strategy is here to be able to tell if the other
  strategies are better or worse than a random choice of pivot.
\item[raremax] \pname{raremax} combines the benefits of removing a
  facet that lacks a popular vertex with the benefit of removing a
  small facet. Table \ref{tab:dbmsPivotNerve} shows that this
  combination is better than the pure strategies \pname{rarevar} and
  \pname{maxsupp}.
\item[rarest] We believe that \pname{rarest} performs well because
  removing facets that lack popular vertices tends to create
  complexes that are close to being a cone.
\item[rarevar] \pname{rarevar} is almost as good as \pname{rarest},
  which is reasonable since \pname{rarest} is a more sophisticated way
  of breaking ties in \pname{rarevar}.
\item[popvar] \pname{popvar} does the opposite of what
  \pname{rarevar} does, so it is reasonable that it does
  poorly.
\item[maxsupp] We believe that \pname{maxsupp} performs better than
  \pname{random} because removing small facets tends to create
  complexes that are close to being a cone.
\item[minsupp] \pname{minsupp} performs worse than \pname{random},
  confirming that it is beneficial to select small pivots.
\end{description}

\begin{table}
\centering
\begin{tabular}{rrrrrrrrr}
$\card{\vertices\Delta}$ & $\card{\facets\Delta}$ &
\pname{random} & \pname{popvar} &
\pname{maxsupp} & \pname{rarest} & \pname{raremax} & \pname{minsup} & \pname{rarevar} \\
\hline
4000 & 20 & 0.07 & 0.06 & 0.06 & {\bf 0.04} & {\bf 0.04} & 0.08 & {\bf 0.04}\\
6000 & 20 & 0.08 & 0.08 & 0.07 & 0.06 & 0.06 & 0.09 & {\bf 0.05}\\
8000 & 20 & 0.09 & 0.09 & 0.08 & {\bf 0.07} & {\bf 0.07} & 0.11 & {\bf 0.07} \\
10000 & 20 & 0.11 & 0.11 & 0.10 & {\bf 0.09} & {\bf 0.09} & 0.15 & 0.10 \\
\hline
4000 & 30  & 4.53 & 4.46 & 2.36 & 1.14 & {\bf 1.06} & 13.34 & 1.24\\
6000 & 30  & 5.96 & 6.77 & 3.58 & 1.47 & {\bf 1.46} & 15.62 & 1.63 \\
8000 & 30  & 7.83 & 8.28 & 4.46 & 1.97 & {\bf 1.90} & 21.17 & 2.26 \\
10000 & 30 & 10.50 & 10.18 & 5.64 & 2.59 & {\bf 2.53} & 27.11 & 3.24 \\
\hline
20 & 4000  &0.19 & 0.18 & 0.16 & {\bf 0.12} & {\bf 0.12} & 0.22 & {\bf 0.12} \\
20 & 6000  & 0.32 & 0.31 & 0.29 & {\bf 0.23} & 0.24 & 0.39 & 0.24\\
20 & 8000  & 0.48 & 0.48 & 0.46 & {\bf 0.37} & {\bf 0.37} & 0.57 & 0.38 \\
20 & 10000 & 0.66 & 0.65 & 0.63 & {\bf 0.52} & 0.54 & 0.74 & 0.54 \\
\hline
30 & 4000  & 5.08 & 4.97 & 3.10 & 1.88 & {\bf 1.76} & 10.68 & 2.06 \\
30 & 6000  & 7.70 & 7.44 & 4.81 & 2.97 & {\bf 2.85} & 16.25 & 3.29 \\
30 & 8000  & 10.16 & 10.02 & 6.46 & 3.99 & {\bf 3.98} & 19.61 & 4.48 \\
30 & 10000 & 12.63 & 13.17 & 8.90 & 5.00 & {\bf 4.98} & 25.35 & 5.58 \\
\hline
100 & 100  & 5.89 & 5.84 & 1.29 & 1.05 & {\bf 0.79} & 40.83 & 1.13\\
120 & 120  & 25.42 & 26.19 & 5.79 & 4.72 & {\bf 3.44} & 173.82 & 5.06 \\
140 & 140  & 108.40 & 119.19 & 22.43 & 18.43 & {\bf 13.03} & 749.66 & 21.34 \\
160 & 160 & 427.46 & 464.28 & 87.13 & 69.92 & {\bf 49.59} & 2767.93 & 78.52 \\
180 & 180  & 1291.26 & 1241.91 & 253.15 & 192.95 & {\bf 143.14} & $>$7200 & 232.63 \\
200 & 200  & 4211.65 & 4137.68 & 769.79 & 614.39 & {\bf 434.43} & $>$7200 & 755.41\\
220 & 220  & $>$7200 & $>$7200 & 1704.65 & 1440.19 & {\bf 980.31} & $>$7200 & 1665.74 \\
240 & 240 & $>$7200 & $>$7200 & 4871.15 & 4022.04 & {\bf 2697.84} & $>$7200 & 4811.50\\
\smallskip
\end{tabular}
\caption{DBMS pivot selection strategies. All times are in seconds.}
\label{tab:dbmsPivotNerve}
\end{table}

\subsection{Variations of the BCRT and DBMS Algorithms}

Table \ref{tab:bcrtPivotWithoutNerve} and Table
\ref{tab:dbmsPivotWithoutNerve} show times for the pivot selection
strategies when the nerve technique from Section
\ref{sec:algImp}/Section \ref{sec:simImp} is turned off. First of all
we observe that turning the nerve technique off does not change the
ranking of the pivot selection strategies. Furthermore, we see that
without nerves the BCRT algorithm is much more sensitive to the number
of vertices while the DBMS technique is much more sensitive to the
number of facets. The nerve technique hides this sensitivity as it
allows to interchange the number of facets and the number of vertices,
so the DBMS algorithm can adjust the input to make the number of
facets less than the number of vertices and vice versa for the BCRT
algorithm.

An alternative to the nerve technique is to use a hybrid approach
where the BCRT algorithm is used for ideals with more facets and the
DBMS algorithm is used for ideals with more vertices. If we compare
the tables we see that the DBMS algorithm with the nerve technique
turned on is faster than the hybrid approach even for ideals with more
facets than vertices, so in this experiment the hybrid approach is
inferior to the nerve technique.

\begin{table}
\centering
\begin{tabular}{rrrrrr}
$\card{V_\Delta}$ & $\card{\facets\Delta}$ & \tt random & \tt popvar & \tt rarevar & \tt popgcd \\
\hline
30 & 10000 &7.50 & {\bf 5.36} & 9.35 & 30.24 \\
10000 & 30 & $>$7200 & {\bf 62.13} & $>$7200 & $>$7200 \\
240 & 240 & $>$7200 & \bf 6961.76 & $>$7200 & $>$7200 \\
\smallskip
\end{tabular}
\caption{BCRT pivot selection strategies without nerves. All times are in seconds.}
\label{tab:bcrtPivotWithoutNerve}
\end{table}

\begin{table}
\centering
\begin{tabular}{rrrrrrrrr}
$\card{\vertices\Delta}$ & $\card{\facets\Delta}$ &
\pname{random} &
\pname{raremax} & \pname{rarest} & \pname{rarevar} & \pname{popvar} &
\pname{maxsupp} & \pname{minsupp} \\
\hline
30 & 10000 & 318.41 & \bf 52.39 & 64.01 & 77.87 & 243.41 & 71.85 & 2382.40 \\
10000 & 30 & 10.64 & \bf 2.58 & 2.66 & 3.26 & 11.30 & 5.97 & 25.95 \\
240 & 240 & $>$7200 & \bf 4184.07 & 6212.34 & $>$7200 & $>$7200 & $>$7200 & $>$7200 \\
\smallskip
\end{tabular}
\caption{DBMS pivot selection strategies without nerves. All times are in seconds.}
\label{tab:dbmsPivotWithoutNerve}
\end{table}

\subsection{Comparison of Euler Characteristic Implementations}

Table \ref{tab:allPrograms} compares several implementations of Euler
characteristic algorithms. This also serves as a comparison of the
algorithms used by these implementations. Evaluating the practicality
of an algorithm as opposed to an implementation is difficult because
quality of implementation has a significant effect on performance yet
quality of implementation cannot easily if at all be measured or
corrected for.

An example of an implementation (as opposed to algorithm) difference
is that Frobby and Macaulay 2 are written in C++ that compiles to
native code while Gap and Sage do not compile to native code due to
the languages that they are written in.\footnote{Sage does have
  modules written in Cython which is similar to Python but that does
  compile to native code. Many components of both Sage and Gap are
  written in native languages such as C and C++. However, that is not
  the case for the Euler characteristic components of Gap and
  Sage. Parts of Macaulay 2 are written in the interpreted Macaulay 2
  language, but the \hps{} code is written in C++.} While the choice
of implementation language can make a significant difference for
performance, the magnitude of the differences in Table
\ref{tab:allPrograms} is so large that choice of language is unlikely
to be the main factor in our estimation.

We can draw the firm conclusion from Table \ref{tab:allPrograms} that
the BCRT and DBMS algorithms presented in this paper can be
implemented to be faster on this set of complexes than all the other
Euler characteristic implementations that we have compared. We have
found no faster Euler characteristic implementations than these, so we
believe that the comparison we have made is the most fair and
informative comparison that can be made using the implementations that
exist today.

From Table \ref{tab:bcrtPivotNerve}, Table \ref{tab:dbmsPivotNerve}
and Table \ref{tab:allPrograms}, we see that the DBMS algorithm is
faster than the BCRT algorithm for all the complexes when both
algorithms use their best pivot selection strategy. So we can
recommend using the DBMS algorithm over the BCRT algorithm, even
though the difference is slight. The nerve technique is vital to make
the DBMS algorithm always be faster -- without it, good performance
could only be reached by implementing both algorithms and choosing
which to use based on the ratio of facets to vertices.

The implementations in \Sage{} and \GAP{} are based on enumerating
faces. Table \ref{tab:data} shows the number of faces of each complex,
and there is a clear trend that the time taken by \Sage{} and \GAP{} is
related to the number of faces of the complex. In contrast the times
for the implementations in \Frobby{} and \Mtwo{} do not have a clear
relationship to the number of faces.

For example nicgraph-8-2 has 2928 times more faces than match-13 does,
and as expected \Sage{} and \GAP{} take much longer to compute the
Euler characteristic of nicgraph-8-2 than of match-13. In contrast
\Frobby{} computes the Euler characteristic of nicgraph-8-2 in 0.06s
while it takes 7.84s for match-13. So nicgraph-8-2 has 2928 times more
faces than match-13 does yet it takes 260 times less time to compute
its Euler characteristic using \Frobby{}. We give this as evidence
that the time taken by the BCRT and DBMS algorithms for Euler
characteristic depends more on the structure of the complex than on
the number of faces of the complex.

We find that the performance of Frobby and \Mtwo{} in this comparison
lends credence to the idea of using algebraic formulations and
implementations for combinatorial problems.

\begin{table}
\centering
\begin{tabular}{lrrrr}
Example&Vertices&Facets&Faces& $\euler\Delta$ \\
\hline
rook-6-6 & 36 & 720 & 13,327 & 185 \\
rook-7-7 & 49 & 5,040 & 130,922 & -204 \\
rook-8-8 & 64 & 40,320 & 1,441,729 & -6,209 \\
\hline
match-9 & 36 & 945 & 2,620 & -28 \\
match-10 & 45 & 945 & 9,496 & -1,216 \\
match-11 & 55 & 10,395 & 35,696 & -936 \\
match-12 & 66 & 10,395 & 140,152 & 12,440 \\
match-13 & 78 & 135,135 & 568,503 & 23,672 \\
\hline
nicgraph-7-2 & 21 & 217 & 1,014,888 & -120 \\
nicgraph-8-2 & 28 & 504 & 166,537,616 & -720 \\
nicgraph-9-2 & 36 & 1,143 &  50,680,432,112& -5,040 \\
\hline
randomv20g100 & 20 & 100 & 86,116 & -25 \\
randomv20g500 & 20 & 500 & 227,792 & 1,166 \\
randomv20g1000& 20 & 1,000& 287,689 & -1,007 \\
randomv25g100 & 25 & 100 & 1,223,224 & -202 \\
randomv25g500 & 25 & 500 & 3,628,979 & -3,815 \\
randomv25g1000& 25 & 1,000& 5,368,430 & 3,666 \\
\smallskip
\end{tabular}
\caption{Characteristics of the examples used in Table \ref{tab:allPrograms}.}
\label{tab:data}
\end{table}

\begin{table}
\begin{threeparttable}
\centering
\begin{tabular}{lrrrrr}
&  \Frobby{} & \Frobby{} &  &  &  \\
Example & DBMS & BCRT & \Sage{} & \Mtwo{} & \GAP{}\\
\hline
rook-6-6 & {\bf 0.01}& {\bf 0.01} & 1.04 & 0.24 & 0.13\\
rook-7-7 & {\bf 0.13} & 0.14 & 12.59 & 3.37 & 3.86\\
rook-8-8 & {\bf 2.43} & 6.39 & $>$223.11\tnote{*} & 58.16 & $>$7200\\
\hline
match-9 & {\bf 0.00} & {\bf 0.00} & 0.16 & 0.13 & 0.08\\
match-10 & {\bf 0.02} & {\bf 0.01} & 0.69 & 0.29 & 0.12\\
match-11 & 0.21 & {\bf 0.15} & 2.90 & 2.41 & 6.97\\
match-12 & {\bf 0.33} & 0.47 & 12.31 & 8.32 & 9.15\\
match-13 & {\bf 7.84} & 11.26 & $>$7200 & 101.49 & 2401.58\\
\hline
nicgraph-7-2 & {\bf 0.00} & {\bf 0.00} & 322.94 & 0.43 &22.41\\
nicgraph-8-2 & {\bf 0.03} & 0.06 & $>$7200 & 10.33 &$>$7200\\
nicgraph-9-2 & {\bf 0.40} & 0.65 & $>$7200 & 306.28 &$>$7200\\
\hline
randomv20f100 & {\bf 0.00} & {\bf 0.00} & 11.37 & 0.07 &0.42\\
randomv20f500 & {\bf 0.01} & {\bf 0.01} & 35.99 & 0.43 & 3.28\\
randomv20f1000 & {\bf 0.02} & {\bf 0.02} & 47.32 & 0.72 & 7.60\\
randomv20f100 & {\bf 0.00} & {\bf 0.00} & 322.80 & 0.37 &58.17\\
randomv20f500 & {\bf 0.02} & 0.04 & $>$7200 & 3.16 & 592.70\\
randomv20f1000 & {\bf 0.02} & 0.04 & $>$7200 & 6.43 & $>$7200\\
\smallskip
\end{tabular}
\begin{tablenotes}
\item [*] Sage reports taking 223.11s on rook-8-8, but the actual time
  was in excess of half an hour. The discrepancy persisted across
  several runs. We do not know how to explain the discrepancy since
  the times that Sage reports is usually in line with external
  measurements.
\end{tablenotes}
\end{threeparttable}
\ \\
\ \\
\caption{Comparison of Euler characteristic implementations. All times are in seconds.}
\label{tab:allPrograms}
\end{table}

\bibliographystyle{plain} \bibliography{references}

\begin{thebibliography}{10}

\bibitem{eulerSheaf}
E.~Bach.
\newblock Sheaf cohomology is {\#}p-hard.
\newblock {\em Journal of Symbolic Computation}, 27(4):429 -- 433, 1999.

\bibitem{bayerMonomialLectureNotes}
Dave Bayer.
\newblock Monomial ideals and duality.
\newblock Never finished draft. See
  \url{http://www.math.columbia.edu/~bayer/vita.html}, 1996.

\bibitem{hseries}
Dave Bayer and Mike Stillman.
\newblock Computation of {H}ilbert functions.
\newblock {\em Journal of Symbolic Computation}, 14(1):31--50, 1992.

\bibitem{bigattiHilbSerComp}
Anna~M. Bigatti.
\newblock Computation of {H}ilbert-{P}oincar{\'e} series.
\newblock {\em Journal of Pure and Applied Algebra}, 119(3):237--253, 1997.

\bibitem{bigattiEtAlHilbSerlg}
Anna~Maria Bigatti, Pasqualina Conti, Lorenzo Robbiano, and Carlo Traverso.
\newblock A ``divide and conquer'' algorithm for {H}ilbert-{P}oincar{\'e}
  series, multiplicity and dimension of monomial ideals.
\newblock In {\em Applied algebra, algebraic algorithms and error-correcting
  codes (San Juan, PR, 1993)}, volume 673 of {\em Lecture Notes in Comput.
  Sci.}, pages 76--88. Springer, Berlin, 1993.

\bibitem{eulerGeom}
Peter B\"{u}rgisser and Felipe Cucker.
\newblock Counting complexity classes for numeric computations i: Semilinear
  sets.
\newblock {\em SIAM J. Comput.}, 33(1):227--260, 2004.

\bibitem{simpcomp}
Felix Effenberger and Jonathan Spreer.
\newblock {\em {\tt simpcomp} - a {\tt GAP} toolkit for simplicial complexes},
  2010.

\bibitem{m2}
Daniel~R. Grayson and Michael~E. Stillman.
\newblock Macaulay 2, a software system for research in algebraic geometry.
\newblock Available at \url{http://www.math.uiuc.edu/Macaulay2/}.

\bibitem{gap}
The~{GAP} Group.
\newblock {\em {GAP} -- Groups, Algorithms, and Programming}, 2008.

\bibitem{nerves}
Branko Gr{\"u}nbaum.
\newblock Nerves of simplicial complexes.
\newblock {\em Aequationes Mathematicae}, 4:63--73, 1970.
\newblock 10.1007/BF01817747.

\bibitem{simplicialSurvey}
Volker Kaibel and Marc~E. Pfetsch.
\newblock Some algorithmic problems in polytope theory.
\newblock In {\em Algebra, Geometry, and Software Systems}, pages 23--47, 2003.

\bibitem{KR97}
D.~A. Klain and G-C. Rota.
\newblock {\em Introduction to geometric probability}.
\newblock Cambridge University Press, 1997.

\bibitem{cca}
Ezra Miller and Bernd Sturmfels.
\newblock {\em Combinatorial Commutative Algebra}, volume 227 of {\em Graduate
  Texts in Mathematics}.
\newblock Springer, 2005.

\bibitem{frobby}
Bjarke~Hammersholt Roune.
\newblock Frobby -- a software system for computations with monomial ideals.
\newblock Available at \url{http://www.broune.com/frobby/}.

\bibitem{rouneSliceIrrDecom}
Bjarke~Hammersholt Roune.
\newblock The slice algorithm for irreducible decomposition of monomial ideals.
\newblock {\em Journal of Symbolic Computation}, 44(4):358--381, April 2009.

\bibitem{S97}
Richard~P. Stanley.
\newblock {\em Enumerative Combinatorics}, volume~1.
\newblock Cambridge University Press, 1997.

\bibitem{sage}
William Stein and David Joyner.
\newblock {SAGE}: Open source mathematics software.
\newblock Available at \url{http://www.sagemath.org/}, 2005.

\bibitem{enumerationSP}
Leslie~G. Valiant.
\newblock The complexity of enumeration and reliability problems.
\newblock {\em SIAM Journal on Computing}, 8(3):410--421, 1979.

\end{thebibliography}

\end{document}